\documentclass{article}

\usepackage{cite}   
\usepackage[numbers]{natbib}
\usepackage{hyperref}
\usepackage{hypernat}
\usepackage{graphicx}
\usepackage{pgfplots}
\usepackage[final]{pdfpages}
\usepackage{url}
\usepackage{verbatim}
\usepackage{listings}
\usepackage{bm}
\usepackage[algo2e, vlined]{algorithm2e}

\SetCommentSty{mycommfont}
\SetKwRepeat{Do}{do}{while}
\SetKw{Or}{or}
\SetKw{And}{and}

\usepackage{subcaption}
\usepackage{lipsum}
\usepackage{multicol}
\usepackage{amsthm}
\newtheorem{definition}{Definition}

\newtheorem{thm}{Theorem}
\newtheorem{lem}[thm]{Lemma}

\usepackage[shortlabels]{enumitem}
\setlist[enumerate]{nosep}

\makeatletter
    \def\realunderscores{\catcode`\_=\active}
    \def\real@ltgtus{%
        \catcode`<=\active
        \catcode`>=\active
        \realunderscores
    }
    {\real@ltgtus
        \gdef<{\futurelet\@let@token\less@than}%
        \gdef>{\futurelet\@let@token\greater@than}%
        \gdef_{\underscore}%
    }
    \DeclareRobustCommand{\underscore}{\ifmmode\sb\else\textunderscore\fi}
    \newcommand{\less@than}{\ifdim\fontdimen4\font=0pt\string<\else
        \leavevmode\mathhexbox13C\fi}
    \newcommand{\greater@than}{\ifdim\fontdimen4\font=0pt\string>\else
        \leavevmode\mathhexbox13E\fi}
    \DeclareRobustCommand{\pcode}{\begingroup
        \real@ltgtus
        \@makeother\&\@makeother\#%
        \@makeother\^\@makeother\%\@makeother\~%
        \@pcode}
    \let\pcodefont\sf
    \def\@pcode#1{{\pcodefont #1}\strut\endgroup}
    \let\code\pcode
\makeatother

\theoremstyle{plain}
\newcommand*{\Note}[1]{\textcolor{red}{\textit{(#1)}}}
\newcommand*{\MLSnote}[1]{\Note{MLS: #1}}
\newcommand*{\New}[1]{\textcolor{blue}{\textit{(#1)}}}
\newcommand*{\WTCnote}[1]{\New{WTC: #1}}

\definecolor{code_bg}{rgb}{0.98,0.98,0.98}
\definecolor{code_hl}{rgb}{1.0,0.98,0.8}

\newcommand*{\hightlightCode}[1]{\setlength{\fboxsep}{0pt}\colorbox{code_hl}{#1}}

\newif\iftighten \tightenfalse
\newif\ifarxiv \arxivtrue

\setlist[enumerate]{style=unboxed,leftmargin=*,topsep=1ex,itemsep=0ex}
\setlist[itemize]{style=unboxed,leftmargin=*,topsep=1ex,itemsep=0ex}
\setlist[description]{style=unboxed,leftmargin=*,topsep=1ex,itemsep=0ex}

\SetKwProg{Fn}{Function}{}{}
\SetKwProg{Class}{Class}{}{}
\LinesNumbered
\DontPrintSemicolon

\makeatletter
\newcommand{\removelatexerror}{\let\@latex@error\@gobble}
\let\old@printtopmatter\@printtopmatter
\def\@printtopmatter{%
  \global\setbox\mktitle@bx=\vbox{\noindent\box\mktitle@bx\par\bigskip}%
  \old@printtopmatter}
\makeatother

\begin{document}

\title{Transactional Composition of Nonblocking Data Structures}

\author{Wentao Cai, Haosen Wen, and Michael L. Scott \\
\texttt{\{wcai6,hwen5,scott\}@ur.rochester.edu} \\
University of Rochester
}
\maketitle
\begin{abstract}
This paper introduces \emph{nonblocking transaction composition}
(NBTC), a new methodology for atomic composition of
nonblocking operations on concurrent data structures. Unlike previous
software transactional
memory (STM) approaches, NBTC leverages the linearizability of
existing nonblocking structures, reducing the number of memory
accesses that must be executed together, atomically, to only one per
operation in most cases (these are typically the linearizing
instructions of the constituent operations).

Our obstruction-free implementation of NBTC, which we call 
\emph{Medley}, makes it easy
to transform most 
nonblocking data structures into transactional counterparts while
preserving their liveness and high concurrency.
In our experiments, Medley outperforms
Lock-Free
Transactional Transform (LFTT), the fastest prior competing methodology,
by 40--170\%.
The marginal overhead of Medley's transactional composition, relative to
separate operations performed in succession, is roughly 2.2$\times$.

For \emph{persistent} data structures, we observe that failure atomicity for
transactions can be achieved ``almost for free'' with epoch-based
\emph{periodic persistence}.  Toward that end, we integrate Medley with
\emph{nbMontage}, a general system for periodically persistent data
structures. 
The resulting \emph{txMontage} provides ACID transactions and achieves
throughput up to two orders of magnitude higher than that of the OneFile
persistent STM system.

\end{abstract}





\section{Introduction}
\label{sec:intro}

Nonblocking concurrent data structures, first explored in the 1970s,
remain an active topic of research today.
In such structures, there is no reachable state of the system
that can prevent an individual operation from making forward progress.
This \emph{liveness} property is highly desirable in multi-threaded
programs that
\iftighten
\else
aim for high scalability and are sensitive to high tail latency caused by
inopportune preemption of resource-holding threads.
\fi

Many multi-threaded systems\iftighten\else
, including those for finance, travel~\cite{minh-iswc-2008}, warehouse
management~\cite{tpc-tpcc-2010}, and databases in
general~\cite{tu-sosp-2013},\fi\ need to compose
operations into transactions that occur in an all-or-nothing fashion
(i.e., atomically).
\iftighten\else Concurrent data structures, however, ensure
atomicity only for individual operations; composing a transaction across
operations requires nontrivial programming effort and introduces high
overhead.  Preserving nonblocking liveness for every transaction is even
more difficult.

\fi
One potential solution can be found in
\emph{software transactional memory} (STM) systems, which convert
almost arbitrary sequential code into speculative
transactions.  Several STM systems provide nonblocking
progress~\cite{herlihy-podc-2003, fraser-thesis-2004,
marathe-transact-2006, marathe-ppopp-2008, tabba-spaa-2009}.
Most instrument each memory access and arrange to
restart operations that conflict at the level of individual loads and
stores.  The resulting programming model is attractive, but the
instrumentation typically imposes 3--10$\times$
overhead~\cite[Sec.\ 9.2.3]{scott-sms-2013}.


Inspired by STM, Spiegelman et
al.~\cite{spiegelman-2016-pldi} proposed \emph{transactional data
structure libraries} (TDSL), which introduce (blocking) transactions
for certain hand-modified concurrent data structures.
By observing that reads need to be tracked only on 
\emph{critical nodes} whose updates may indicate semantic conflicts, 
TDSL reduces read set
size and achieves better performance than general STMs.

Herlihy and Koskinen~\cite{herlihy-ppopp-2008} proposed
\emph{transactional boosting}, a (blocking) methodology that allows
an STM system to incorporate operations on existing concurrent data
structures.
Using a system of \emph{semantic locks} (e.g., with one lock per key
in a mapping), transactions arrange to
execute concurrently so long as their boosted operations are logically
independent, regardless of low-level conflicts.
A transaction that restarts due to a semantic
conflict (or to a low-level conflict outside the boosted code)
will roll back any already-completed boosted operations by
performing explicitly identified \emph{inverse operations}%
\iftighten\else. An
\mbox{\code{insert(k,v)}} operation, for example, would be rolled back by
performing \code{remove(k)}\fi.
Transactional boosting leverages the
potential for high concurrency in existing data structures, but is
intrinsically lock-based, and is not fully general:
operations on a single-linked FIFO queue, for example, have no obvious inverse.

In work concurrent to TDSL, Zhang et al.~\cite{zhang-spaa-2016} proposed
the \emph{Lock-Free Transactional Transform} (LFTT), a nonblocking
methodology to compose nonblocking data structures, based on the
observation that only \iftighten\else certain nodes---those critical to transaction
semantics---\fi really matter in conflict management.
Each operation on an LFTT structure \emph{publishes}, on
every critical node, a description of the transaction of which it is a
part, so that conflicting transactions can see and \emph{help} each
other%
\iftighten\else. A \code{remove(7)} operation, for example, would publish a
description of its transaction on the node in its structure with
key~7\fi.
Initially, LFTT supported only static transactions, whose
constituent operations were all known in advance.  Subsequently, LaBorde
et al.~\cite{laborde-pmam-2019} proposed a \emph{Dynamic Transactional
  Transform} (DTT) that generalizes LFTT to dynamic transactions
(specified as lambda expressions). Concurrently, Elizarov et
al.~\cite{elizarov-ppopp-2019} proposed LOFT, which is similar to LFTT
but avoids 
incorrectly repeated helping.

Unfortunately, as in transactional boosting, the
need to identify critical nodes tends to limit
LFTT and DTT to data structures representing sets and mappings.
DTT's publishing and helping mechanisms also require that the
``glue'' code between operations be
fully reentrant (to admit concurrent execution by helping
threads~\cite{laborde-pmam-2019}) and may result in redundant work
when conflicts arise.
Worse, for read-heavy workloads,
LFTT and DTT require readers to be \emph{visible} to writers, introducing
metadata updates that significantly increase contention in the cache
coherence protocol.

In our work, we propose \emph{NonBlocking Transaction Composition}
(NBTC), a new methodology that can create transactional versions of a
wide variety of concurrent data structures
while preserving nonblocking progress and incurring significantly lower
overhead than traditional STM\@.
The intuition behind NBTC is that in already nonblocking
structures, only \emph{critical memory accesses}---for the most part,
the linearizing load and compare-and-swap (CAS) instructions---need to occur
atomically, while most pre-linearization memory accesses can safely be
executed as they are encountered, and post-linearization accesses can be
postponed until after the transaction commits.

In comparison to STM, NBTC significantly reduces the number of
memory accesses that must be instru\-mented---typically to only one per
constituent operation.  Unlike transactional boosting and transactional
transforms, NBTC brings the focus back from semantics to
low-level memory accesses, thereby enabling mechanical transformation
of existing structures and accommodating almost arbitrary
abstractions---much more than sets and mappings.
NBTC also supports dynamic transactions, invisible readers, and
non-reentrant ``glue'' code between the operations of a transaction.
The one requirement for compatibility is that the linearization points
of constituent operations must be \emph{immediately
  identifiable}: each operation must be able to tell when it has
linearized at run time, without performing any additional shared-memory
accesses.  Most nonblocking structures in the literature appear to meet
this requirement.




To assess the practicality of NBTC, we have built an ob\-struc\-tion-free
implementation, \emph{Medley}, that uses
a variant of Harris et al.'s 
multi-word CAS~\cite{harris-disc-2002} to 
execute the critical memory accesses of each transaction atomically,
eagerly resolving conflicting transactions as they are discovered.
Using Medley, we have created NBTC versions of Michael and Scott's
queue~\cite{michael-podc-1996}, Fraser's
skiplist~\cite{fraser-thesis-2004}, the rotating skiplist of Dick et
al.~\cite{dick-rotating-2016}, Michael's chained hash
table~\cite{michael-spaa-2002}, and Natarajan and Mittal's binary search
tree~\cite{natarajan-ppopp-2014}.  All of the transformations were
straightforward.

In the traditional language of database
transactions~\cite{haerder-compsur-1983}, Medley provides isolation and
consistency.  Building on recent work on persistent memory, we have also
integrated Medley with the \emph{nbMontage} system of
\citet{cai-disc-2021} to create a system, \emph{txMontage}, that
provides failure atomicity and durability as well---i.e., full ACID
transactions.
Specifically, we leverage the \emph{epoch} system of nbMontage, which
divides time into coarse-grain temporal intervals and recovers, on
failure, to the state of a recent epoch boundary.  By folding a check of
the epoch number into its multi-word CAS, txMontage ensures that
operations of the same transaction always linearize in the same epoch,
thereby obtaining failure atomicity and durability ``almost for free.''

Summarizing contributions:

\begin{itemize}[leftmargin=1em,parsep=.5ex plus .5ex]
    \item
        (Section~\ref{sec:nbtc})
        We introduce \emph{nonblocking transaction composition} (NBTC), a
        new methodology with which to compose the operations of
        nonblocking data structures.
    \item
        (Section~\ref{sec:impl})
        Deploying NBTC, we implement \emph{Medley}, a general system for
        transactional nonblocking structures. Medley's easy-to-use
        API and mechanical transform make it easy to convert
        compatible nonblocking structures to transactional form.
    \item
        (Section~\ref{sec:persistence})
        We integrate Medley with nbMontage to create \emph{txMontage},
        providing not only transactional isolation and consistency, but
        also failure atomicity and durability.
    \item
        (Section~\ref{sec:proof})
        We argue that using NBTC and Medley, transactions composed of
        nonblocking structures are nonblocking and strictly serializable. We
        also argue that transactions with txMontage provide a persistent
        variant of strict serializability analogous to the buffered
        durable linearizability of \citet{izraelevitz-disc-2016}.
    \item
        (Section~\ref{sec:exp})
        We present performance results, confirming that Medley imposes
        relatively modest overhead and scales to large numbers of
        threads. Specifically, Medley outperforms LFTT by $1.4\times$ to
        $2.7\times$ and outperforms TDSL and the OneFile nonblocking
        STM~\cite{ramalhete-dsn-2019} system by an order of
        magnitude. On persistent memory, txMontage outperforms
        nonblocking persistent STM by two orders of magnitude.
\end{itemize}

\section{Nonblocking Transaction Composition}
\label{sec:nbtc}

\emph{Nonblocking transaction composition} (NBTC) is a new methodology
that fully leverages the linearizability of nonblocking data structure
operations.
NBTC obtains strict serializability by atomically performing only the
\emph{critical memory accesses} of composed operations.
It
supports a large subset of the nonblocking data structures in the
literature (characterized more precisely below), preserving the high
concurrency and nonblocking liveness of the transformed structures. 

\subsection{NBTC Composability}
The key to NBTC composability is the \emph{immediately identifiable
linearization point}. Specifically:

\begin{definition}
A data structure operation has an \emph{immediately
  identifiable linearization point} if:
\begin{enumerate}

  \item statically, we can identify every instruction that may
  potentially serve as the operation's linearization point. Such
  an instruction must be a load for a read-only operation or a
  compare-and-swap (CAS) for an update operation;

  \item dynamically, after executing a potentially linearizing
  instruction, we can determine whether it was indeed the linearization
  point. A linearizing load must be identified before the operation
  returns; a linearizing CAS must be identified without performing any
  additional shared-memory accesses.
\end{enumerate}

\end{definition}

There can be more than one potential linearization point in the code of
an operation, but only one of them will constitute the linearization
point in any given invocation.


\begin{definition}\label{def:nbtc_composable}
A nonblocking data structure is \emph{NBTC-com\-pos\-able} if
each of its operations has an immediately
identifiable linearization point.
\end{definition}

\iftighten\else
While it may be possible to relax this definition, the current version
accommodates a very large number of existing nonblocking structures.
\fi




\subsection{The Methodology}
\label{sec:methodology}

It is widely understood that most nonblocking operations comprise a
``planning'' phase and a ``cleanup'' phase, separated by a linearizing
instruction~\cite{timnat-ppopp-2014,friedman-pldi-2020}.  Executing
the planning phase does not commit the operation to success; cleanup,
if needed, can be performed by any thread. The basic strategy in NBTC
is to perform the planning for all constituent operations of the
current transaction, then linearize all those operations together,
atomically, and finally perform all cleanup.  Our survey of existing
data structures and composition patterns reveals two principle
complications with this strategy.


 
The first complication involves the notion of a \emph{publication
  point}, where an operation may 
become visible to other threads but not yet linearize.  Because
publication can alter the behavior of other threads, it
must generally (like a linearization point) remain speculative until the entire
transaction is ready to commit.  An example can be seen in the binary
search tree of Natarajan and Mittal~\cite{natarajan-ppopp-2014}, where
an update operation $o$ may perform a CAS that publishes its intent to
linearize soon but not quite yet.  After this publication point, either
$o$ itself or any other update that encounters the publication notice
may attempt to
linearize $o$ (in the interest of performance, a read operation will
ignore it).  Notably, CAS instructions that serve to help other
(already linearized) operations, without revealing the nature of the
current operation, need not count as publication.

The second complication arises when a transaction, $t$, performs two or
more operations on the same data structure and one of the later
operations (call it $o_2$) depends on the outcome of an earlier
operation (call it $o_1$).  Here the thread executing $t$ must proceed
as if $o_1$ has completed, but other threads must ignore it.  If $o_1$
requires cleanup (something that NBTC will normally delay
until after transaction commit), $o_2$ may need to help $o_1$ before it
can proceed, while other transactions should not even be aware of
$o_1$'s existence.


Both complicating cases can be handled by introducing the notion of a
\emph{speculation interval} in which CAS instructions must be
completed together for an operation to take effect as part of a
transaction.  This is similar to the \emph{CAS executor} phase in a
\emph{normalized} nonblocking data structure~\cite{timnat-ppopp-2014},
but not the same, largely due to the second complication.
For an operation that becomes visible before its linearization point,
it suffices to include in the speculation interval all CAS operations
between the publication and linearization points, inclusive.  For an
operation $o_2$ that needs to see an earlier operation $o_1$ in the
same transaction, it suffices to track the transaction's writes and to
start $o_2$'s speculation interval no later than the first instruction
that accesses a location written by~$o_1$.

\begin{definition}\label{def:critical}
A bit more precisely, we say
\begin{itemize}
\item
    A CAS instruction in operation $o$ of thread $t$ in history $H$ is
    \emph{benign} if there is no extension $H^\prime$ of $H$ such that
    $t$ executes no more instructions in $H^\prime$ and yet $o$
    linearizes in $H^\prime$ nonetheless.
\item
    The first CAS instruction of $o$ that is not benign is $o$'s
    \emph{publication point} (this will often be the same as its
    linearization point).
\item
    The \emph{speculation interval} of $o$ begins either at the
    publication point or at the first instruction that sees a value
    speculatively written by some earlier operation in the same
    transaction (whichever comes first) and extends through $o$'s
    linearization point.
\item
    A load in a read-only operation is \emph{critical} if it is the
    immediately identifiable linearization point of the operation.  A
    CAS in an update operation is critical if it lies in the
    speculation interval.
\end{itemize}
\end{definition}

Without loss of generality, we assume that all updates to shared
memory (other than initialization of objects not yet visible to other
threads) are effected via CAS.

Given these definitions, the NBTC methodology is straightforward: 
To atomically execute a set of operations on NBTC-composable data
structures, we transform every operation such that (1) instructions
prior to the speculation interval and non-critical instructions in the
speculation interval are executed on the fly as a transaction
encounters them; (2) critical instructions are executed in a
speculative fashion, so they will take effect, atomically, only on
transaction commit; and (3) instructions after the speculation
interval are postponed until after the commit.

\section{The Medley System}
\label{sec:impl}

To illustrate NBTC, we have written a system, \emph{Medley}, that
(1)~instruments critical instructions, executes them speculatively,
and commits them atomically using \emph{M-compare-N-swap}, our variant
of the multi-word CAS of \citet{harris-disc-2002};
\linebreak[3]
(2)~identifies and eagerly resolves transaction conflicts; and
\mbox{(3) delays} non-crit\-i\-cal clean\-up until transaction commit.

\subsection{API}
\label{sec:api}

Figure~\ref{fig:api} summarizes Medley's API\@.  Using this API, we
transform an NBTC-composable data structure into a transactional
structure as follows:

\lstdefinestyle{myCustomCppStyle}{
  language=C++,
  stepnumber=1,
  tabsize=1,
  showspaces=false,
  showstringspaces=false,
  xleftmargin=0pt,
  basicstyle=\scriptsize\ttfamily\selectfont,
  numberstyle=\ttfamily\tiny,
  escapeinside=||,
  keywordstyle=\bfseries\color{red!40!black},
  commentstyle=\itshape\color{green!40!black},
  columns=fullflexible,
}

\begin{figure}
\begin{lstlisting}[style=myCustomCppStyle,xleftmargin=12pt,numbersep=4pt,numbers=left]
template <class T> class CASObj { // Augmented atomic object
  T nbtcLoad();
  bool nbtcCAS(T expected, T desired, bool linPt, bool pubPt);
  /* Regular atomic methods: */
  T load();  void store(T desired);  bool CAS(T expected, T desired);
};
class Composable {  // Base class of all transactional objects
  template <class T> void addToReadSet(CASObj<T>*,T); // Register load
  void addToCleanups(function); // Register post-critical work
  template <class T> T* tNew(...);     // Create a new block
  template <class T> void tDelete(T*); // Delete a block
  template <class T> void tRetire(T*); // Epoch-based safe retire
  TxManager* mgr; // Tx metadata shared among Composables
  struct OpStarter { OpStarter(TxManager*); } // RAII op starter
};
class TxManager { // Manager shared among composable objects
  void txBegin(); // Start a transaction
  void txEnd();   // Try to commit the transaction
  void txAbort(); // Explicitly abort the transaction
  void validateReads(); // Optional validation for opacity
};
struct TransactionAborted : public std::exception{ };
\end{lstlisting}
  \vspace{-3ex}
  \captionsetup{justification=centering}
  \caption{C++ API of Medley for transaction composition. }
  \vspace{-2ex}
  \label{fig:api}
\end{figure}

\begin{enumerate}
\item
  Replace critical loads and CASes with \code{nbtcLoad} and
  \code{nbtcCAS}, respectively.  Fields to which such accesses are
  made should be declared using the \code{CASObj} template.
\item
  Invoke \code{add\-To\-Read\-Set} for the critical load in a read
  operation, recording the address and the loaded value.
\item
  Register each operation's post-critical work via
  \code{add\-To\-Clean\-ups}.
\item
  Replace every \code{new} and \code{delete} with \code{tNew} and
  \code{tDelete}.  Replace every \code{retire} (for safe memory
  reclamation---SMR) with \code{tRetire}.
\item
  Declare an \code{OpStarter} object at the beginning of each
  operation.
\end{enumerate}

\code{CASObj<T>} augments each CAS-able 64-bit word (e.g.,
\code{atomic<Node*>}) with additional metadata bits for speculation
tracking (details in Section \ref{sec:mcns}).
It provides specialized load and CAS operations, as well as the usual methods
of \code{atomic<T>}. To dynamically identify the speculation interval,
\code{nbtcCAS} takes two extra arguments, \code{linPt} and
\code{pubPt}, that indicate whether this call, should it succeed, will
constitute its operation's linearization or/and publication point.
In a similar vein, \code{add\-To\-Read\-Set} can be called after an
\code{nbtcLoad} to indicate (after inspecting the return value) that
this was (or is likely to have been) the linearizing load of a read-only
operation, and should be tracked for validation at commit time.


\code{Composable} is a base class for transactional objects.
\iftighten
\else
It provides a variety of NBTC-related methods, including
\fi
support for
\emph{safe memory reclamation} (SMR), used to ensure that nodes are not
reclaimed until one can be certain that no references remain among the
private variables of other threads.  Our current implementation of SMR uses
epoch-based reclamation \cite{fraser-thesis-2004, hart-jpdc-2007,
  mckenney-ols-2001}.
\iftighten\else
For the sake of generality, \code{Composable} also provides an API for
transactional boosting, which can be used to incorporate lock-based
operations into Medley transactions (at the cost, of course, of
nonblocking progress).  We do not discuss this mechanism further in this
paper.
\fi

The \code{TxManager} class manages transaction metadata and provides
methods to initiate, abort, and complete a transaction.  A
\code{TxManager} instance is shared among all \code{Composable}
instances intended for use in the same transactions.  In each operation
call, the manager distinguishes (via \code{OpStarter()}) whether
execution is currently inside or outside a transaction. If outside, all
transactional instrumentation is
elided; if inside, instrumentation proceeds as specified by the NBTC
methodology.

Given that nonblocking operations can execute safely in any
reachable state of the system, there is usually no need to stop the
execution of a doomed-to-abort transaction as soon as a conflict
arises---i.e., to guarantee \emph{opacity}~\cite{guerraoui-ppopp-2008}.
In exceptional cases (e.g., when later operations of a transaction
cannot be called with certain combinations of parameters, or when
aborts are likely enough that delaying them may compromise
performance), the \code{validate\-Reads} method can be used to determine
whether previous reads remain correct.

To illustrate the use of Medley, Figure~\ref{fig:ht_example}
highlights lines of code in Michael's nonblocking hash
table~\cite{michael-spaa-2002} that must be modified for NBTC;
Figure~\ref{fig:tx_example} then shows an example transaction that
modifies two hash tables. In a real application, the \code{catch}
block for \code{TransactionAborted} would typically loop back to the
beginning of the transaction code to try again, possibly with
additional code to avoid livelock (e.g., via backoff
or hints to the underlying scheduler). In contrast to STM systems,
Medley does not instrument the intra-transaction ``glue'' code between
data structure operations.  This code is always executed as regular
code outside a transaction and should always be data-race free; if it
has side effects, the \code{catch} block (written by the programmer)
for aborted transactions should compensate for these before the
programmer chooses to retry or give up.

\begin{figure}
\begin{lstlisting}[style=myCustomCppStyle,xleftmargin=12pt,numbersep=4pt,numbers=left]
class MHashTable |\hightlightCode{: public Composable}| {
struct Node { K key; V val; |\hightlightCode{CASObj<Node*>}| next; };
// from p, find c >= k; |\hightlightCode{nbtcLoad and tRetire}| may be used
bool find(|\hightlightCode{CASObj<Node*>}|* &p, Node* &c, Node* &n, K k); 
optional<V> get(K key) {
  |\hightlightCode{OpStarter starter(mgr);}| |\hightlightCode{CASObj<Node*>*}| prev = nullptr; 
  Node *curr, *next; optional<V> res = {};
  if (find(prev,curr,next,key)) res = curr->val;
  |\hightlightCode{addToReadSet(prev,curr);}|
  return res;
}
optional<V> put(K key, V val) { // insert or replace if key exists
  |\hightlightCode{OpStarter starter(mgr);}|
  |\hightlightCode{CASObj<Node*>*}| prev = nullptr; optional<V> res = {};
  Node *newNode = |\hightlightCode{tNew<Node>}|(key, val), *curr, *next;
  while(true) {
    if (find(prev,curr,next,key)) { // update
      newNode->next.store(curr);
      if (curr->next.|\hightlightCode{nbtcCAS}|(next,mark(newNode)|\hightlightCode{,true,true}|)) {
        res = curr->val;
        |\hightlightCode{auto cleanup = []()}|{
          if (prev->CAS(curr,newNode)) |\hightlightCode{tRetire}|(curr);
          else find(prev,curr,next,key); 
        };
        |\hightlightCode{addToCleanups(cleanup);}| // execute right away if not in tx
        break;
      }
    } else { // key does not exist; insert
      newNode->next.store(curr);
      if (prev->|\hightlightCode{nbtcCAS}|(curr,newNode|\hightlightCode{,true,true}|)) break;
    }
  }
  return res;
}};
\end{lstlisting}
\vspace{-3ex}
\captionsetup{justification=centering}
\caption{Michael's lock-free hash table example (Medley-related parts highlighted).}
\label{fig:ht_example}
\vspace{-2ex}
\end{figure}

\begin{figure}
\begin{lstlisting}[style=myCustomCppStyle,xleftmargin=12pt,numbersep=4pt,numbers=left]
void doTx(MHashTable* ht1, MHashTable* ht2, V v, K a1, K a2) {
  TxManager* mgr=ht1->mgr; assert(mgr==ht2->mgr); 
  try { // transfer `v' from account `a1' in `ht1' to `a2' in `ht2'
    mgr->txBegin();
    V v1 = ht1->get(a1); V v2 = ht2->get(a2); 
    if (!v1.hasValue() or v1.value() < v) mgr->txAbort();
    ht1->put(a1, v1.value() - v); ht2->put(a2, v + v2.valueOr(0));
    mgr->txEnd();
  } catch (TransactionAborted) { /* transaction aborted */ }
}
\end{lstlisting}
\vspace{-3ex}
\captionsetup{justification=centering}
\caption{Transaction example on Michael's hash table.}
\label{fig:tx_example}
\vspace{-2ex}
\end{figure}

\subsection{M-Compare-N-Swap}
\label{sec:mcns}

To execute the critical memory accesses of each transaction
atomically, we employ a software-emulated \emph{M-compare-N-swap}
(MCNS) operation that builds on the double-compare-single-swap (RDCSS) and
multi-word CAS (CASN) of \citet{harris-disc-2002}.  Each transaction
maintains a \emph{descriptor} that contains a read set, a write set,
and a 64-bit triple of thread ID, serial number, and status, as shown
in Figure~\ref{fig:desc_casobj}.  Descriptors are pre-allocated on a
per-thread basis within a \code{TxManager} instance, and are reused
across transactions. A status can be \code{InPrep} (initial state),
\code{InProg} (ready to commit), \code{Committed} (after validation
succeeds when \code{InProg}), or \code{Aborted} (explicitly by another
thread when \code{InPrep} or due to failed validation). 

Each originally 64-bit word at
which a critical memory access may occur is augmented with a 64-bit
counter, together comprising a 128-bit \code{CASObj}. Each critical CAS
installs a pointer to its descriptor in the \mbox{\code{CASObj}} and increments
the counter; upon commit or abort, the descriptor is
uninstalled and the counter incremented again.  We leverage 128-bit CAS
instructions on the x86 to change the original word and the
counter together, atomically.  The counter is odd when \code{CASObj}
contains a pointer to a descriptor and even when it is a real value.

\begin{figure}
\begin{lstlisting}[style=myCustomCppStyle,xleftmargin=12pt,numbersep=4pt,numbers=left]
struct Desc {
  map<CASObj* addr,{uint64 val,cnt}>* readSet; 
  map<CASObj* addr,{uint64 oldVal,cnt,newVal}>* writeSet;
  atomic<uint64> status;//63..50 tid 49..2 serialNumber 1..0 status
  enum STATUS { InPrep=0, InProg=1, Committed=2, Aborted=3 };
};
struct CASObj { atomic<uint128> val_cnt; };
\end{lstlisting}
\vspace{-3ex}
\captionsetup{justification=centering}
\caption{Descriptor and \code{CASObj} structures.}
\label{fig:desc_casobj}
\vspace{-2ex}
\end{figure}

\begin{figure}
\begin{lstlisting}[style=myCustomCppStyle,xleftmargin=12pt,numbersep=4pt,numbers=left]
void TxManager::txBegin() {
  desc->readSet->clear(); desc->writeSet->clear();
  status.store((status.load() & ~3) + 4);
}
T CASObj::nbtcLoad() {
retry:
  {val,cnt} = val_cnt.load();
  if (cnt % 2) { // is descriptor
    if (val == desc) {
      startSpeculativeInterval(); |\label{code:read_speculative}|
      return desc->writeSet[this].newVal; |\label{code:read_own_desc}|
    } else val->tryFinalize(this, {val,cnt});
    goto retry; // until object has real value
  }
  ... /* Record `this' and `cnt' to be added to readSet */|\label{code:record_cnt}|
  return val;
}
void Composable::addToReadSet(CASObj<T>* obj, T val) {
  ... /* Retrieve `cnt' by `obj` */
  mgr->readSet[obj] = {val,cnt}; |\label{code:add_to_read}|
}
bool CASObj::nbtcCAS(T expected,T desired,bool linPt,bool pubPt){
retry:
  {val,cnt} = val_cnt.load();
  if (cnt % 2) { // is descriptor
    if (val != desc) { // not own descriptor
      val->tryFinalize(this, {val,cnt}); 
      goto retry; // until object has real value
    }
    startSpeculativeInterval(); |\label{code:write_speculative1}|
  } else if (val != expected) return false;
  if (pubPt) startSpeculativeInterval(); |\label{code:write_speculative2}|
  if (inSpeculativeInterval()) { // Is critical CAS
    desc->writeSet[this] = {val,cnt,desired}; |\label{code:update_writeset}|
    bool ret = true;
    if (!(cnt % 2)) ret = this->CAS({val,cnt},{desc,cnt+1});
    if (!ret) desc->writeSet.remove(this);
    if (linPt and ret) endSpeculativeInterval(); |\label{code:spec_end}|
    return ret;
  } else return CAS(expected, desired);
}
\end{lstlisting}
\vspace{-3ex}
\captionsetup{justification=centering}
\caption{Pseudocode for installing phase of MCNS.}
\label{fig:install_pseudo}
\vspace{-2ex}
\end{figure}

\begin{figure}
\begin{lstlisting}[style=myCustomCppStyle,xleftmargin=12pt,numbersep=4pt,numbers=left]
bool Desc::stsCAS(uint64 d, STATUS expected, STATUS desired) {
  d = d & ~3; return status.CAS(d + expected, d + desired);
}
bool Desc::setReady(){return stsCAS(status.load(),InPrep,InProg);}
bool Desc::commit(uint64 d){return stsCAS(d,InProg,Committed);}
bool Desc::abort(uint64 d){return stsCAS(d,d & 1,Aborted);}
void Desc::tryFinalize(CASObj* obj, uint128 var) { |\label{code:try_complete}|
  uint64 d = status.load();
  if (obj->val_cnt.load() != var) // ensure d indicates right tx |\label{code:status_recheck}|
    return;
  if (d & 3 == InPrep) {
    abort(d);
    uint64 newd = status.load();
    if (newd & ~3 != d & ~3) return; // serial number mismatch
    d = newd;
  }
  if (d & 3 == InProg) {
    if (validateReads(d)) commit(d); |\label{code:helper_commit}|
    else abort(d);
  }
  uninstall(status.load()); |\label{code:helper_uninstall}|
}
bool Desc::validateReads() {
  for (e:*readSet) 
    if ({e.val,e.cnt} != e.addr->load()) return false; |\label{code:validate_read}|
  return true;
}
void Desc::uninstall(uint64 d) { 
  if (d % 3 == Committed)
    for (e:*writeSet) 
      e.addr->CAS({this,e.cnt+1}, {e.newVal,e.cnt+2}); |\label{code:uninstall_commit}|
  else // Aborted
    for (e:*writeSet) 
      e.addr->CAS({this,e.cnt+1}, {e.oldVal,e.cnt+2}); |\label{code:uninstall_abort}|
} 
struct TxManager {
  threadLocal vector<Function> cleanups, allocs;
  threadLocal Desc* desc;
  void txAbort() { |\label{code:txabort_begin}|
    uint64 d = desc->status.load();
    desc->abort(d); 
    desc->uninstall(d); |\label{code:abort_uninstall}|
    for (f:allocs) f(); // undo tNew |\label{code:dealloc}|
    throw TransactionAborted();
  }
  void txEnd() {
    if (!desc->setReady()) txAbort();
    else {
      uint64 d = desc->status.load();
      if (!desc->validateReads()) desc->abort(d); |\label{code:validate}| 
      else if (d & 3 == InProg) desc->commit(d); |\label{code:owner_commit}|
      d = desc->status.load();
      if (d & 3 == Committed) {
        desc->uninstall(d); |\label{code:commit_uninstall}|
        for (f:cleanups) f(); |\label{code:cleanup}|
      } else txAbort();
    }
  }|\label{code:txend_end}|
};
\end{lstlisting}
\vspace{-3ex}
\captionsetup{justification=centering}
\caption{Pseudocode of methods that finalize
  transactions.}
\label{fig:finalize_pseudo}
\vspace{-2ex}
\end{figure}  

Each instance of MCNS proceeds through phases that
install descriptors, finalize status, and uninstall descriptors.  The
first two phases are on the critical path of a data structure operation.
A new transaction initializes metadata in its
descriptor (at \code{txBegin}): it clears the read and write sets,
increments the serial number, and resets the status to
\code{InPrep}.  The installing phase then occurs over the
course of the transaction:  Each critical load
records its address, counter, and
value in the read set.  Each critical CAS
records its address, old counter, old value, and desired new value
in the write set; it then installs a pointer to the descriptor in the
\code{CASObj}.  Pseudocode for the installing phase appears in
Figure~\ref{fig:install_pseudo}.

To spare the programmer the need to reason about counters,
\code{nbtcLoad} makes a record of
its $\langle$counter, object$\rangle$ pair
(line~\ref{code:record_cnt} in Fig.~\ref{fig:install_pseudo});
\code{addToReadSet} then adds this pair (and the specified 
\code{CASObj}) to the transaction's read set (line~\ref{code:add_to_read}).

When a thread encounters its own descriptor, \code{nbtcLoad} returns
the speculated value from the write set (line~\ref{code:read_own_desc}).
Likewise, \code{nbtcCAS} updates the write entry
(line~\ref{code:update_writeset}).  Such encounters automatically
initiate the speculation interval
(lines~\ref{code:read_speculative}, \ref{code:write_speculative1},
and~\ref{code:write_speculative2}),
which then extends through the linearization point of the
current operation (line~\ref{code:spec_end}).

If an operation encounters the descriptor of some other thread, it gets
that descriptor out of the way by calling \code{tryFinalize}
(Fig.~\ref{fig:finalize_pseudo}).  This method aborts the associated
transaction if the descriptor is \code{InPrep}, helps complete the
commit if \code{InProg}, and in all cases uninstalls the descriptor from
the \code{CASObj} in which it was found.
%
%
Similar actions occur when a thread is forced to abort or reaches the
end of its transaction and attempts to commit
(lines~\ref{code:txabort_begin}--\ref{code:txend_end}).
Whether helping or acting on its own behalf, a thread performing an MCNS
must verify that the descriptor is still responsible for the
\code{CASObj} through which it was discovered
(line~\ref{code:status_recheck}) and (if committing) that
the values in the read set are still valid (line~\ref{code:validate_read}).
After CAS-ing the \code{status} to \code{Committed} or \code{Aborted},
the thread uninstalls the descriptor from all associated \code{CASObj}s,
replacing pointers to the descriptor with the appropriate updated
values (lines~\ref{code:uninstall_commit} and~\ref{code:uninstall_abort}).
Once uninstalling is complete, the owner thread calls cleanup routines
(line~\ref{code:cleanup}) for a commit or deallocates \code{tNew}-ed
blocks (line~\ref{code:dealloc}) for an abort.

Our design adopts invisible readers and eager contention management
for efficiency and simplicity.  Eager contention management admits the
possibility of livelock---transactions that repeatedly abort each
other---and therefore guarantees only obstruction freedom. Lazy
(commit-time) contention management along with some total order of
descriptor installment might allow us to preserve lock freedom for
structures that provide it~\cite{spear-ppopp-2009}, but would
significantly complicate the tracking and retrieving of uncommitted
changes, and would not address starvation, which may be a bigger
problem than livelock in practice;
we consider these implementation choices orthogonal to the
effectiveness of the NBTC methodology, and defer them to future work.

\section{Persistent Memory}
\label{sec:persistence}

Transactions developed, historically, in the database community;
transactional memory (TM) adapted them to in-memory structures in
multithreaded programs.  The advent of cheap, low-power,
byte-addressable nonvolatile memory (NVM) presents the opportunity to merge
these two historical threads in a way that ideally leverages NBTC\@.
Specifically, where TM aims to convert sequential code to
thread-safe parallel code, NBTC assumes---as in the database
world---that we are already in possession of efficient thread-safe
structures and we wish to combine their operations atomically and
durably.  Given this assumption, it seems appropriate (as described at
the end of Sec.~\ref{sec:api}) to assume that the programmer is
responsible for the ``glue'' code between operations, and to focus on
the atomicity and durability of the composed operations.

\subsection{Durable Linearizability}

 
On machines with volatile caches, data structures in NVM will generally
be consistent after a crash only if programs take pains to issue
carefully chosen write-back and fence instructions.  To characterize
desired behavior, Izraelevitz et al.~\cite{izraelevitz-disc-2016}
introduced \emph{durable linearizability} as a correctness criterion for
persistent structures.  A structure is durably linearizable if it is
linearizable during crash-free execution and its long-term history
remains linearizable when crash events are elided.
Equivalently~\cite{friedman-ppopp-2018}, each operation should persist
between its invocation and response, and the order of persists should
match the linearization order.

Many durably linearizable nonblocking data structures have been designed
in recent years~\cite{friedman-ppopp-2018, chen-atc-2020,
  zuriel-oopsla-2019, fatourou-spaa-2019}.  Several groups have also
proposed methodologies by which existing nonblocking structures can be
made durably linearizable~\cite{izraelevitz-disc-2016,
  friedman-pldi-2020, friedman-pldi-2021}.
Other groups have developed persistent STM systems, but most
of these have been lock-based~\cite{volos-asplos-2011, coburn-asplos-2011,
  correia-spaa-2018, liu-asplos-2017}.
OneFile~\cite{ramalhete-dsn-2019} and QSTM~\cite{beadle-qstm-2020} are,
to the best of our knowledge, the only nonblocking persistent STM
systems.  OneFile serializes transactions using a global sequence
number, eliminating the need for a read set and improving read
efficiency, but introducing the need for invasive data structure
modifications and a 128-bit wide CAS\@.  QSTM employs a global
persistent queue for active transactions, avoiding the need for wide
CAS and invasive structural changes, but with execution that remains
inherently serial.

\subsection{Lowering Persistence Overhead}

Unfortunately, write-back and fence instructions tend to have high
latency.  Given the need for operations to persist before returning,
durable linearizability appears to be intrinsically expensive.
Immediate persistence for STM introduces additional overhead, as
metadata for transaction concurrency control must also be eagerly
written back and fenced.

To move high latency instructions off the application's critical path,
Izrael\-evitz et al.~\cite{izraelevitz-disc-2016} introduced the notion of
\emph{buffered} durable linearizability (BDL).  By allowing a modest
suffix of pre-crash execution to be lost during post-crash recovery (so
long as the overall history remains linearizable), BDL allows write-back
and fence instructions to execute in batches, off the application's
critical path.  Applications that need to ensure persistence before
communicating with the outside world can employ a \code{sync} operation,
reminiscent of those in traditional file systems and databases.

First proposed in the context of the Dal\'{i} persistent hash
table~\cite{nawab-disc-2017}, \emph{periodic persistence} was subsequently
adopted by nbMontage~\cite{cai-disc-2021}, a general-purpose system to
create BDL versions of existing nonblocking structures.  The nbMontage
system divides wall-clock time into ``epochs''
and persists operations in a batch at the end of each epoch.
In the wake of a crash in epoch $e$, the system recovers all structures
to their state as of the end of epoch $e-2$.
To maximize throughput in the absence of crashes, nbMontage also
distinguishes between data that are semantically significant
(a.k.a.\,``payloads'') and data that are merely performance enhancing
(e.g., indices); the latter can be kept in DRAM and rebuilt during
recovery.  As an example, the payloads of a mapping are simply a pile of
key-value pairs; the associated hash table, tree,
or skiplist resides in transient DRAM.  The payloads of a queue
are $\langle$serial number, item$\rangle$ pairs.

To ensure that post-crash recovery always reflects a consistent state of
each structure, every nbMontage operation is forced to linearize in the
epoch with which its payloads have been labeled.  Operations that take
``too long'' to complete may be forced to abort and start over.
The nbMontage system as a whole is lock free; \code{sync} is actually
wait free.

\subsection{Durable Strict Serializability}

Linearizability, of course, is not suitable for transactions, which must
remain speculative until all operations can be made
visible together.  STM systems typically provide strict
serializability instead: transactions in a crash-free history appear
to occur in a sequential order that respects real time (if
$A$ commits before $B$ begins, then $A$ must serialize
before $B$)~\cite[Sec.\ 3.1.2]{scott-sms-2013}.  For a persistent
version of NBTC, we need to accommodate crashes.


Like \citet{izraelevitz-disc-2016}, we assume a full-system crash
failure model: data structures continue to exist after a crash, but are
accessed only by new threads---the old threads disappear.  Under this model:

\begin{definition}
An execution history $H$ displays \emph{durable strict
  serializability} (DSS) if it is strictly serializable when crash events
are elided.
\end{definition}

\noindent
Like durable linearizability, this definition requires all work
completed before a crash to be visible after the crash.  The buffered
analogue is similar:

\begin{definition}
An execution history $H$ displays \emph{buffered durable strict
  serializability} (BDSS) if there exists a happens-before--consistent cut
of each inter-crash interval such that $H$ is strictly serializable
when crash events are elided \emph{along with the post-cut suffix of
each inter-crash interval}.
\end{definition}



\subsection{Merging Medley with nbMontage}

The epoch system of nbMontage provides a natural mechanism with which to
provide failure atomicity and durability for Medley transactions: if
operations of the same transaction always occur in the same epoch, then
they will be recovered (or lost) together in the wake of a crash.
%
%
Building on this observation, we merge the two systems to create
\emph{txMontage}.  Payloads of all operations in a given transaction
are labeled with the same epoch number.  That number is then validated
along with the rest of the read set during MCNS commit, ensuring that
the transaction commits in the expected epoch.  While nbMontage itself
is quite complex, this one small change is all that is required to
graft it (and all its converted persistent data structures) onto
Medley: persistence comes into transactions ``almost for free.''

\renewcommand{\proofname}{Proof (sketch)}

\section{Correctness}
\label{sec:proof}

In this section, we argue that histories comprising well-formed
Medley transactions are strictly serializable, that Medley is obstruction
free, and that txMontage provides buffered durable strict
serializability.

\begin{definition}
  \label{def:tx}
A Medley transaction is \emph{well-formed} if
\begin{enumerate}
\item it starts with \code{txBegin} and ends with \code{txEnd},
  optionally with \code{txAbort} in between;
\item it contains operations of NBTC-transformed data structures; and
\item all other intra-transaction code is nonblocking 
  and free from any side effects not managed by
  handlers for the \code{TransactionAborted} exception.
\end{enumerate}
\end{definition}

\subsection{Strict Serializability}
\label{sec:ss}

\begin{lem}\label{lemma:mcns}
  At the implementation level
  \iftighten\else (operating on the array of words that
  comprises system memory)\fi, \code{nbtcLoad}, \code{nbtcCAS},
  \code{tryFinalize}, \code{txAbort}, and \code{txEnd} (MCNS) are
  linearizable operations.
\end{lem}
\begin{proof}
Follows directly from \citet{harris-disc-2002}.
Their RDCSS compares (without changing) only a single location, and
their CASN supports the update of \emph{all} touched words, but the
proofs adapt in a straightforward way.
In particular, as in RDCSS, an unsuccessful \code{tryFinalize} or
\code{txEnd} can linearize on a
(failed) validating read or a failed CAS of its \code{status} word.
A \code{tryFinalize} or \code{txEnd} whose \code{status} CAS is
successful linearizes ``in the past,'' on the first of its validating
reads.
\iftighten\else(Ironically, this means that MCNS, at the implementation
level, does not have an immediately identifiable linearization point.)\fi
\end{proof}

\begin{lem}\label{lemma:interleaving}
  In any history in which transaction $t$ performs an \code{nbtcLoad} or
  \code{nbtcCAS} operation $x$ on \code{CASObj} $o$, and in which $t$'s
  \code{txEnd} operation $y$ succeeds, no \code{tryFinalize} or
  \code{txEnd} for a different transaction that modifies $o$ succeeds
  between $x$ and $y$.
\end{lem}
\begin{proof}
    Suppose the contrary, and call the transaction with the
    conflicting \code{tryFinalize} or \code{txEnd} $u$.  If $u$'s
    \code{nbtcCAS} of $o$ occurs between $x$ and $y$, it will abort and
    uninstall $t$'s descriptor, or cause read validation to fail in $y$, 
    contradicting the assumption that
    $t$'s \code{txEnd} succeeds.  If $u$'s \code{nbtcCAS} of $o$ occurs
    before $x$, then $x$ will abort and uninstall $u$'s descriptor,
    contradicting the assumption that $u$'s \code{tryFinalize} or
    \code{txEnd} succeeds after $x$.
\end{proof}

\begin{thm}
\label{thm:ss}
Histories comprising well-formed Medley transactions are strictly serializable.
\end{thm}
\begin{proof}
In an NBTC-transformed data structure, all critical memory accesses will
be performed using \code{nbtcLoad} or \code{nbtcCAS}.
These will be followed, at some point, by a call to \code{txEnd}.
If that call succeeds, no conflicting \code{tryFinalize} or \code{txEnd}
succeeds in the interim, by Lemma~\ref{lemma:interleaving}.
This in turn implies that our Medley history is equivalent to a sequential
history in which each operation takes effect at the \code{nbtcLoad} or
\code{nbtcCAS} corresponding to the linearization point of the original
data structure operation, prior to NBTC transformation.  Moreover, all
operations of the same transaction are contiguous in this sequential
history---that is, our Medley history is strictly serializable.
%
%
%
\end{proof}

\subsection{Obstruction Freedom}
\begin{thm}
When used to build well-formed transactions that retry on abort, Medley
is obstruction free.
\end{thm}
\begin{proof}
In any reachable system state, if one thread continues to execute while
others are paused, every \code{nbtcLoad} or \code{nbtcCAS} that
encounters a conflict will first finalize (commit or abort) the
encountered descriptor, uninstall it, and install its own descriptor. If
the thread encounters its own descriptor, a \code{nbtcLoad} will return
the speculated value and a \code{nbtcCAS} will update the write set if
the argument matches the previous new value in the write set.  In either
case, the MCNS will make progress.  If it eventually aborts, it may
repeat one round of a brand new MCNS which, with no newly introduced
contention, must commit.
\end{proof}

\subsection{Buffered Durable Strict Serializability}
\begin{thm}
\label{thm:pp}
Histories comprising well-formed txMontage transactions exhibit buffered
durable strict serializability.
\end{thm}
\begin{proof}
Each transaction reads the current epoch, $e$, in \code{txBegin}.
It then validates this epoch number during MCNS commit.
Per Lemma~\ref{lemma:mcns}, this MCNS must linearize
inside $e$. With nbMontage-provided failure atomicity of all
operations in the same epoch, the theorem trivially holds.
\end{proof}

\section{Performance Results}
\label{sec:exp}

As noted in Section~\ref{sec:intro}, we have used Medley to create NBTC
versions of Michael and Scott’s queue~\cite{michael-podc-1996}, Fraser's
skiplist~\cite{fraser-thesis-2004}, the rotating skiplist of Dick et
al.~\cite{dick-rotating-2016}, Michael's chained hash
table~\cite{michael-spaa-2002}, and Natarajan and Mittal's binary search
tree~\cite{natarajan-ppopp-2014}.  All of the transformations were
straightforward.
In this section we report on the performance of Medley and txMontage
hash tables and skiplists, comparing them to various alternatives from the
literature.

Specifically, we tested the following transient systems:
\begin{description}
\item[Medley --] as previously described (hash table and skip list)
\item[OneFile --] transient version of the lock-free STM of
  Ramalhete et al.~\cite{ramalhete-dsn-2019} (hash table and skip list)
\item[TDSL --] transactional data structure library of Spiegelman et
al.~\cite{spiegelman-2016-pldi} (authors' skiplist only)
\item[LFTT --] lock-free transactional transform of Zhang et
al.~\cite{zhang-spaa-2016} (authors' skiplist only)
\end{description}

We also tested the following persistent systems:
\begin{description}
\item[txMontage --] Medley + nbMontage (hash table and skiplist)
\item[POneFile --] persistent version of
  OneFile~\cite{ramalhete-dsn-2019} (hash table and skiplist)
\end{description}


\subsection{Experimental Setup}

We report throughput for hash table and skiplist microbenchmarks and for
skiplists used to run a subset of TPC-C~\cite{tpc-tpcc-2010}.  We also
measure latency for skiplists.  All code will be made publicly available
prior to conference publication.

All tests were conducted on a Linux 5.3.7 (Fedora 30) server with two
Intel Xeon Gold 6230 processors. Each socket has 20 physical cores and
40 hyperthreads, totaling 80 hyperthreads. Threads in all experiments
were pinned first one per core on socket 0, then on the extra
hyperthreads of that socket, and then on socket 1.  Each socket has 6
channels of 32 GB DRAMs and 6 channels of 128 GB Optane DIMMs. We mount
NVM from each socket as an independent ext4 file system. In all
experiments, DRAM is allocated across the two sockets according to
Linux's default policy; in persistent data structures, only NVM on
socket 0 is used, in direct access (DAX) mode.  In all cases, we report the
average of three trials, each of which runs for 30 seconds.


Our throughput and latency microbenchmark begins by pre-loading the
structure with 0.5\,M key-value pairs, drawn from a key space of 1\,M
keys.  Both keys and values are 8-byte integers.  In the benchmarking
phase, each thread composes and executes transactions comprising 1 to 10
operations each.  Operations (on uniformly random keys) are chosen among
\code{get}, \code{insert}, and \code{remove} in a ratio specified as a
parameter (0:1:1, 2:1:1, or 18:1:1 in our experiments).

In OneFile, we use a sequential chained hash table parallelized
using STM\@.  In Medley, we use an NBTC-transformed version of Michael's
lock-free hash table~\cite{michael-spaa-2002}. Each table has 1\,M
buckets.
In OneFile and TDSL, skiplists are derived from Fraser's STM-based
skiplist~\cite{fraser-thesis-2004}. In LFTT and
Medley, they are derived from Fraser's CAS-based nonblocking
skiplist~\cite{fraser-thesis-2004}. Each skiplist has up to 20 levels.

For TPC-C, we are limited by the fact that Fraser's skiplists do not
support range queries.  Following the lead of \citet{yu-vldb-2014}
\iftighten\else in their experiments with DBx1000\fi \cite{yu-vldb-2014}, we
limit our experiments to TPC-C's \code{newOrder} and \code{payment}
transactions\iftighten\else, which we perform in a 1:1 ratio\fi.
These are the dominant transactions in the benchmark; neither
performs a range query.  


\subsection{Throughput (Transient)}
\label{sec:micro}

\ifarxiv

\begin{figure*}
    \strut\hfill
    \subfloat[get:insert:remove
    0:1:1]%
        {\includegraphics[width=.32\textwidth]%
            {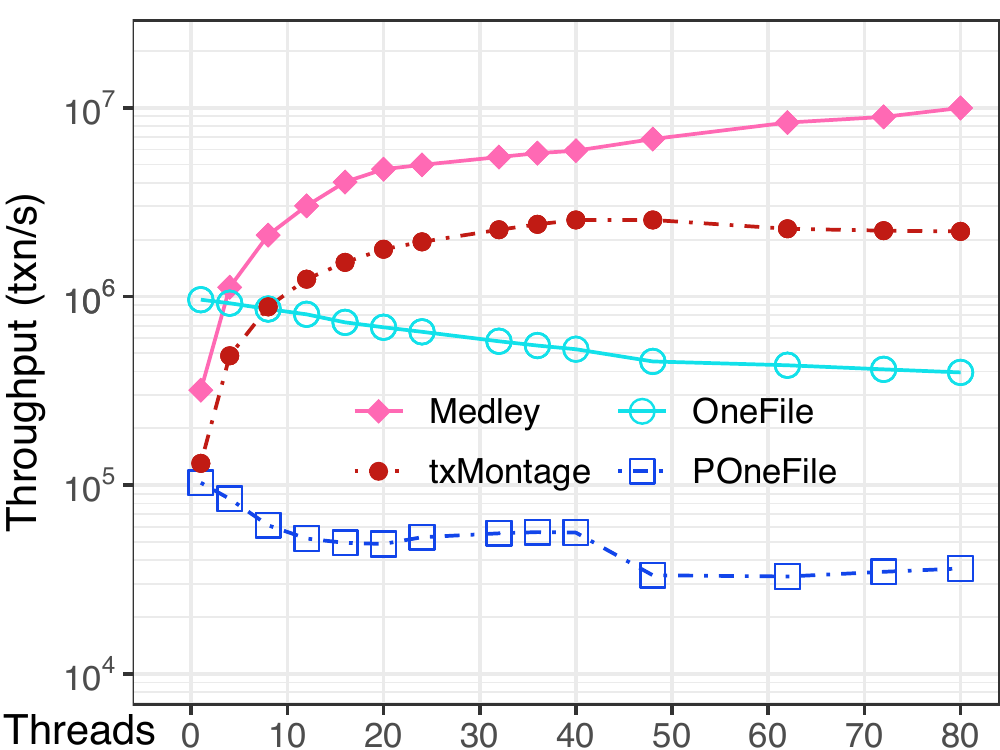}%
        \vspace{-1ex}%
        }
    \hfill
    \subfloat[get:insert:remove
    2:1:1]%
        {\includegraphics[width=.32\textwidth]%
            {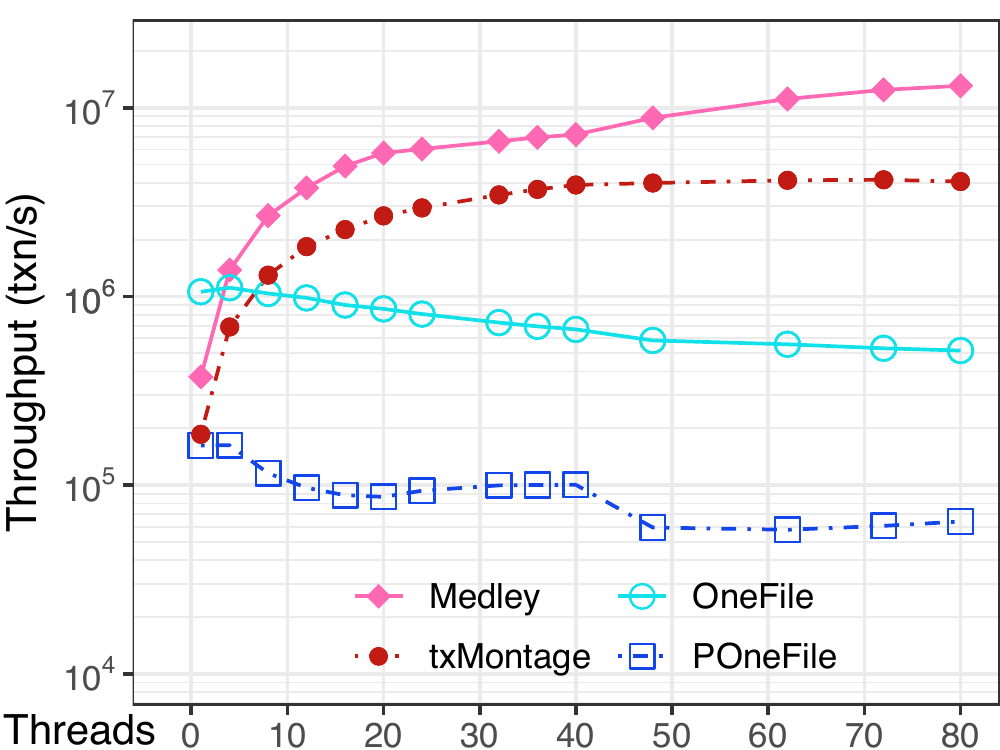}%
        \vspace{-1ex}%
        }
    \hfill
    \subfloat[get:insert:remove
    18:1:1]%
        {\includegraphics[width=.32\textwidth]%
            {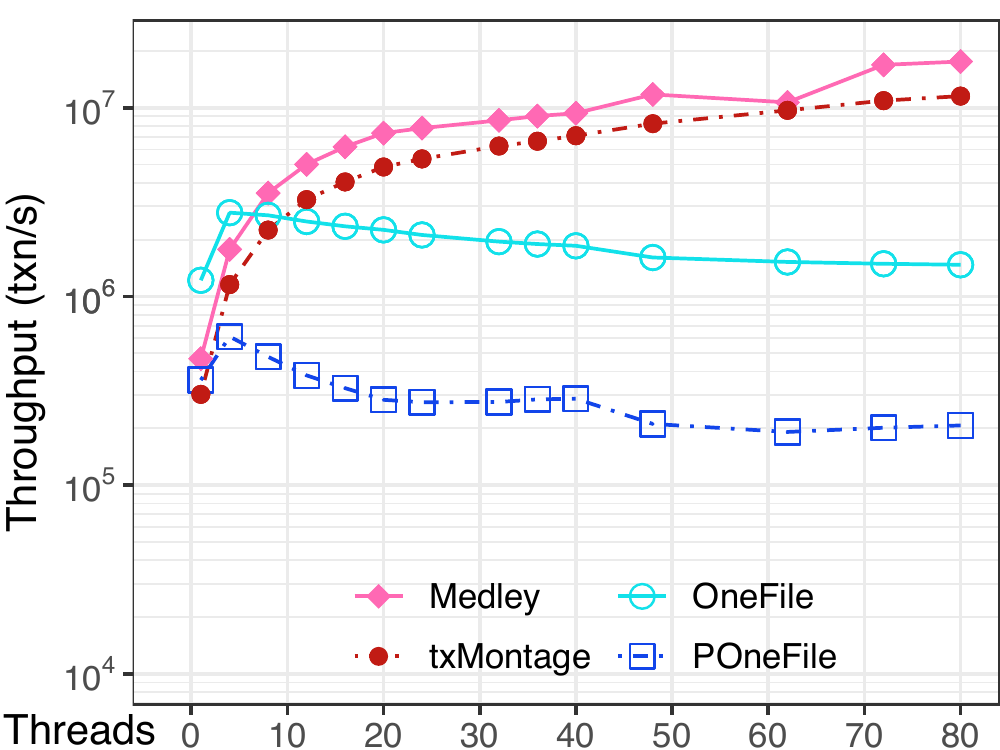}
        \vspace{-1ex}
        }
    \hfill\strut
    \vspace{-2ex}%
    \caption{Throughput of transactional hash tables (log $Y$
      axis).}
    \vspace{-1ex}%
    \label{fig:throughput-hts}
\end{figure*}

\begin{figure*}
    \strut\hfill
    \subfloat[get:insert:remove
    0:1:1]%
        {\includegraphics[width=.32\textwidth]%
            {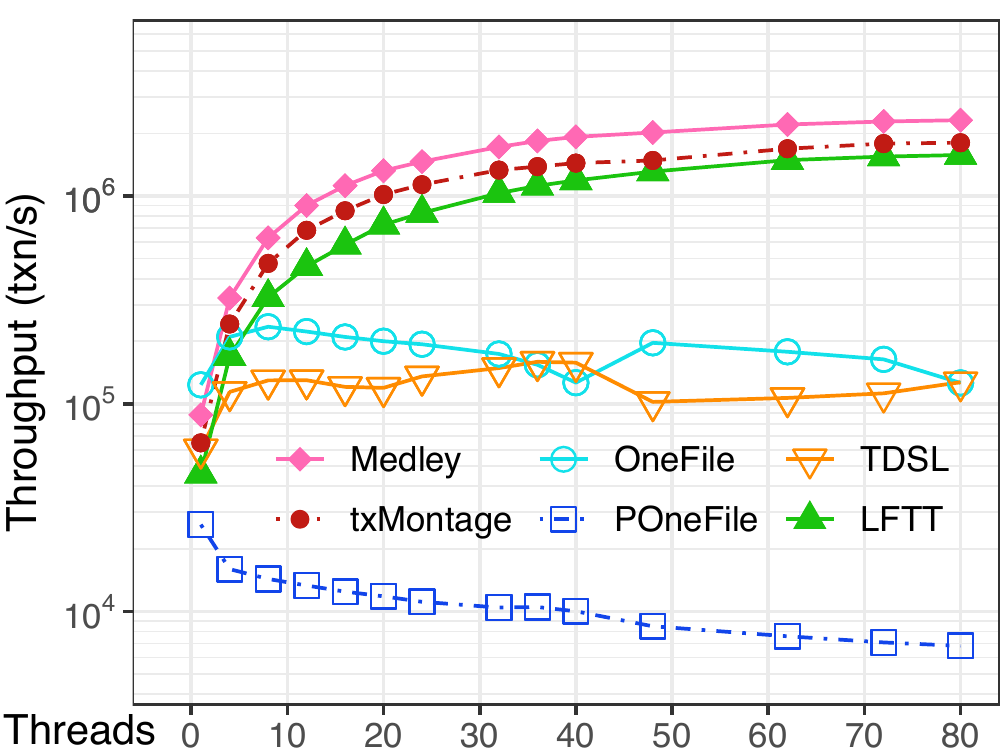}%
        \vspace{-1ex}%
        }
    \hfill
    \subfloat[get:insert:remove
    2:1:1]%
        {\includegraphics[width=.32\textwidth]%
            {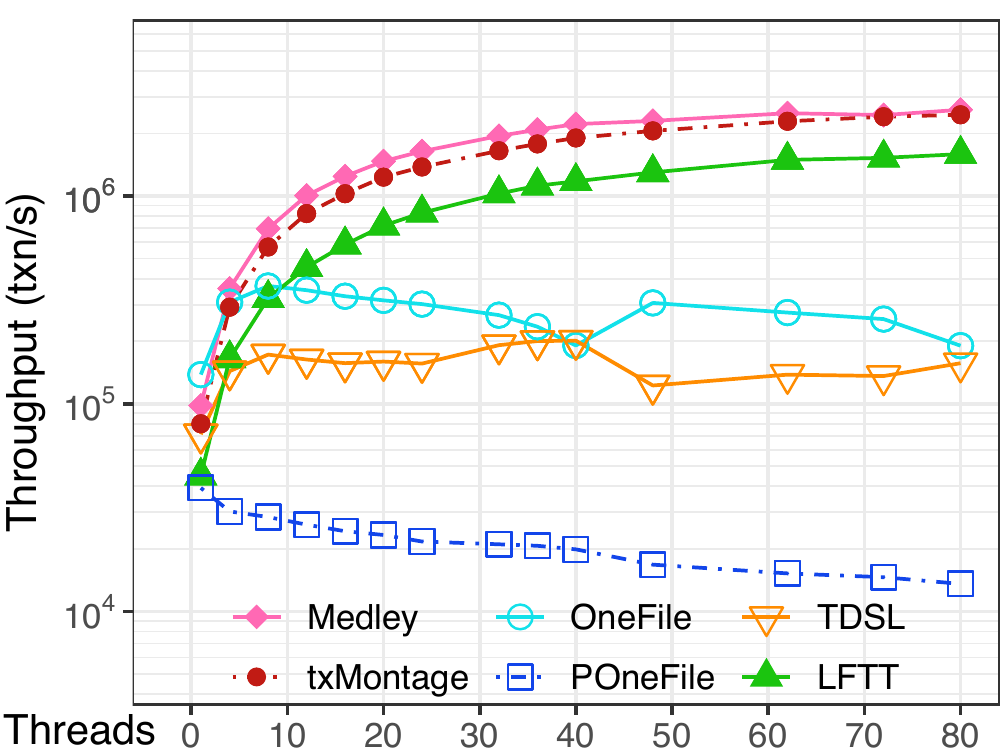}%
        \vspace{-1ex}%
        }
    \hfill
    \subfloat[get:insert:remove
    18:1:1]%
        {\includegraphics[width=.32\textwidth]%
            {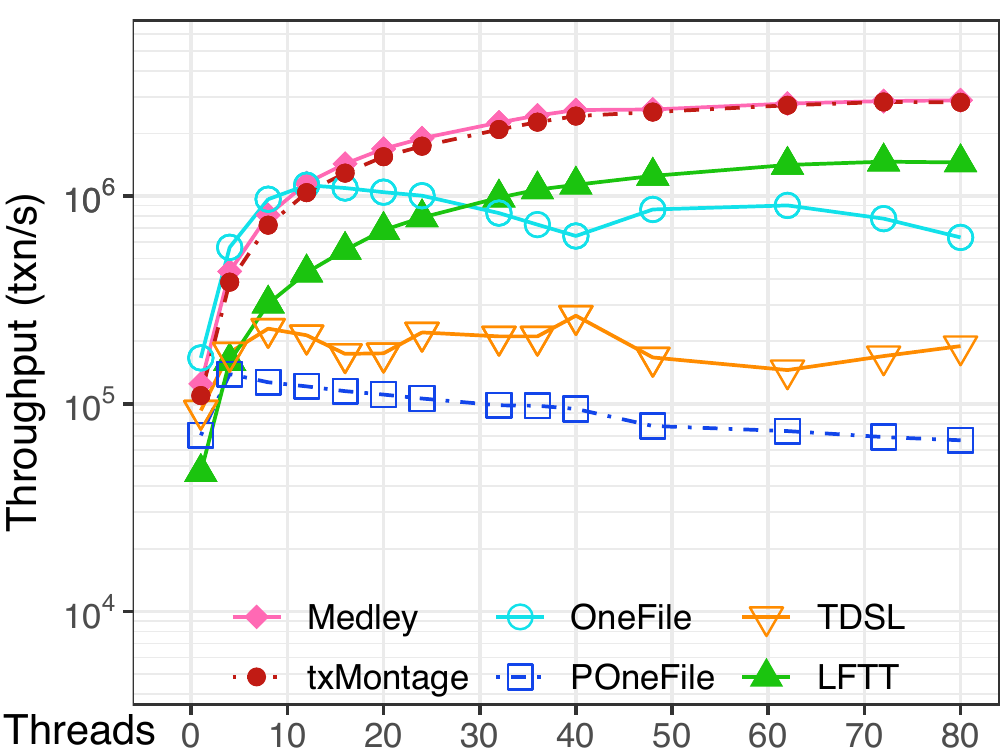}
        \vspace{-1ex}
        }
    \hfill\strut
    \vspace{-2ex}%
    \caption{Throughput of transactional skiplists (log $Y$
      axis).}
    \vspace{-1ex}%
    \label{fig:throughput-sls}
\end{figure*}

\else
\begin{figure*}
    \strut\hfill
    \subfloat[get:insert:remove
    0:1:1]%
        {\includegraphics[width=.3\textwidth]%
            {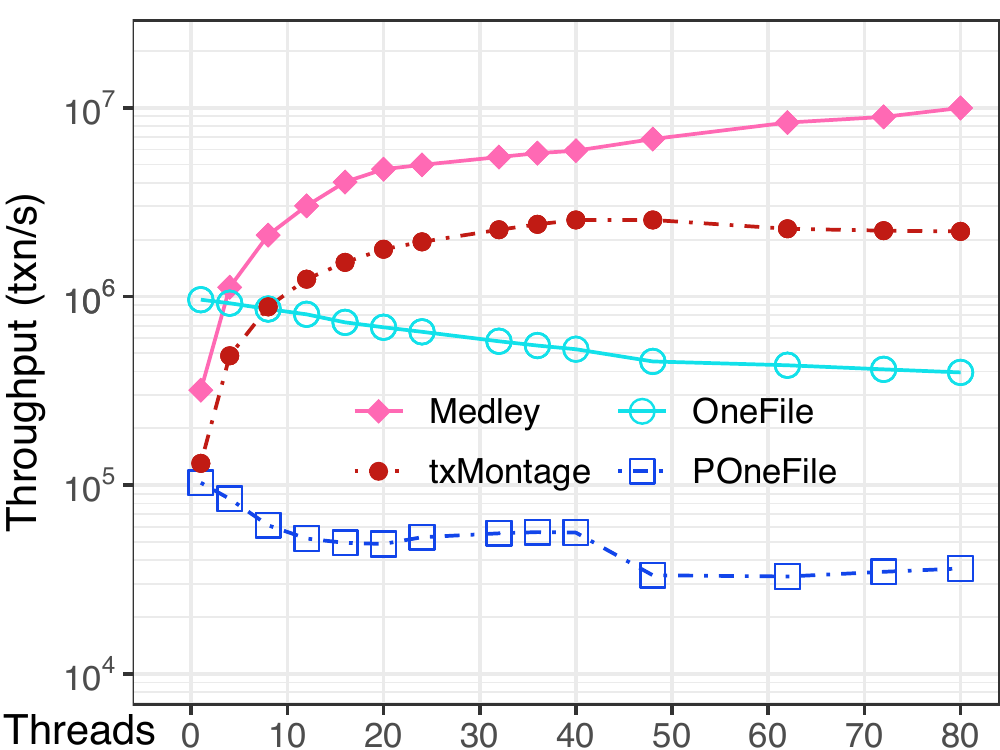}%
        \vspace{-1ex}%
        }
    \hfill
    \subfloat[get:insert:remove
    2:1:1]%
        {\includegraphics[width=.3\textwidth]%
            {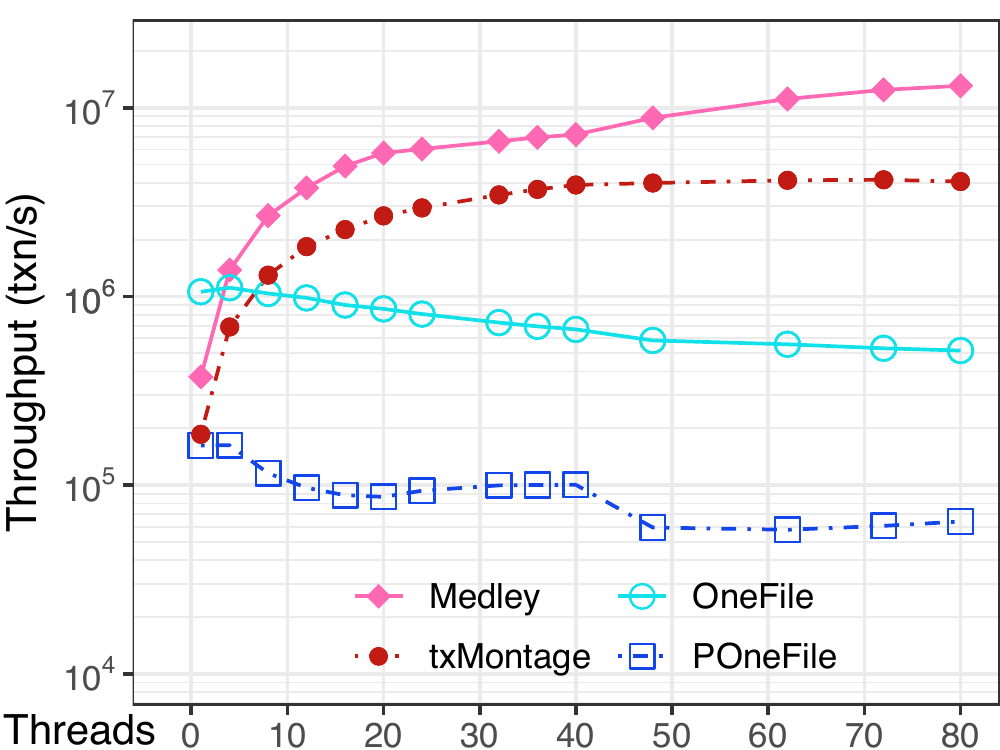}%
        \vspace{-1ex}%
        }
    \hfill
    \subfloat[get:insert:remove
    18:1:1]%
        {\includegraphics[width=.3\textwidth]%
            {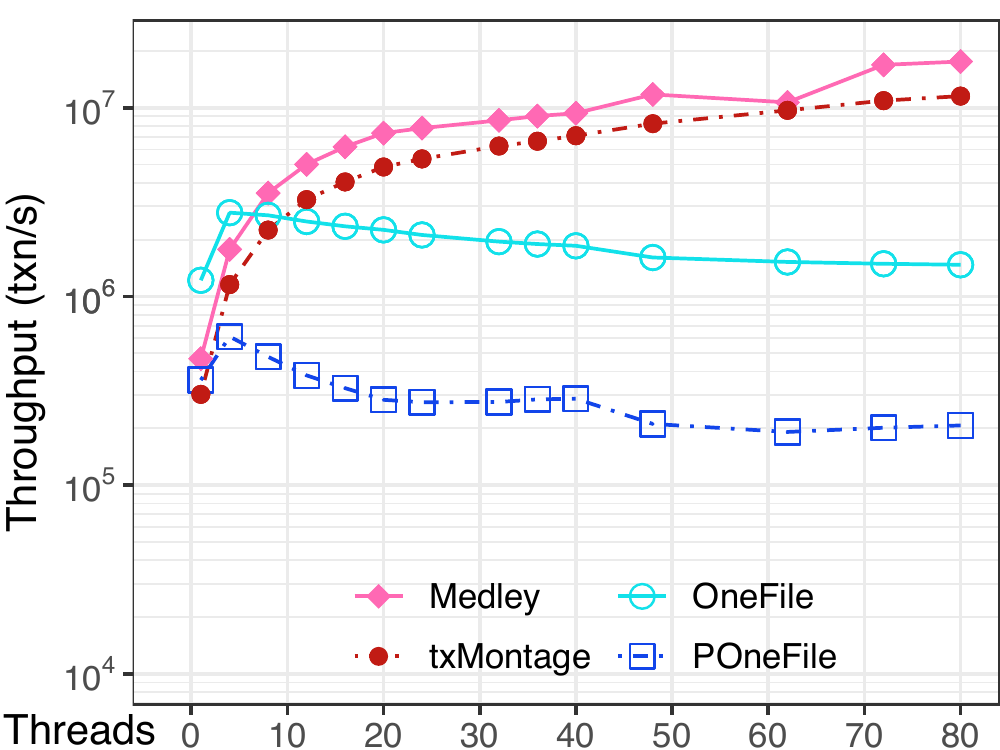}
        \vspace{-1ex}
        }
    \hfill\strut
    \vspace{-2ex}%
    \caption{Throughput of transactional hash tables (log $Y$
      axis).}
    \vspace{-1ex}%
    \label{fig:throughput-hts}
\end{figure*}

\begin{figure*}
    \strut\hfill
    \subfloat[get:insert:remove
    0:1:1]%
        {\includegraphics[width=.3\textwidth]%
            {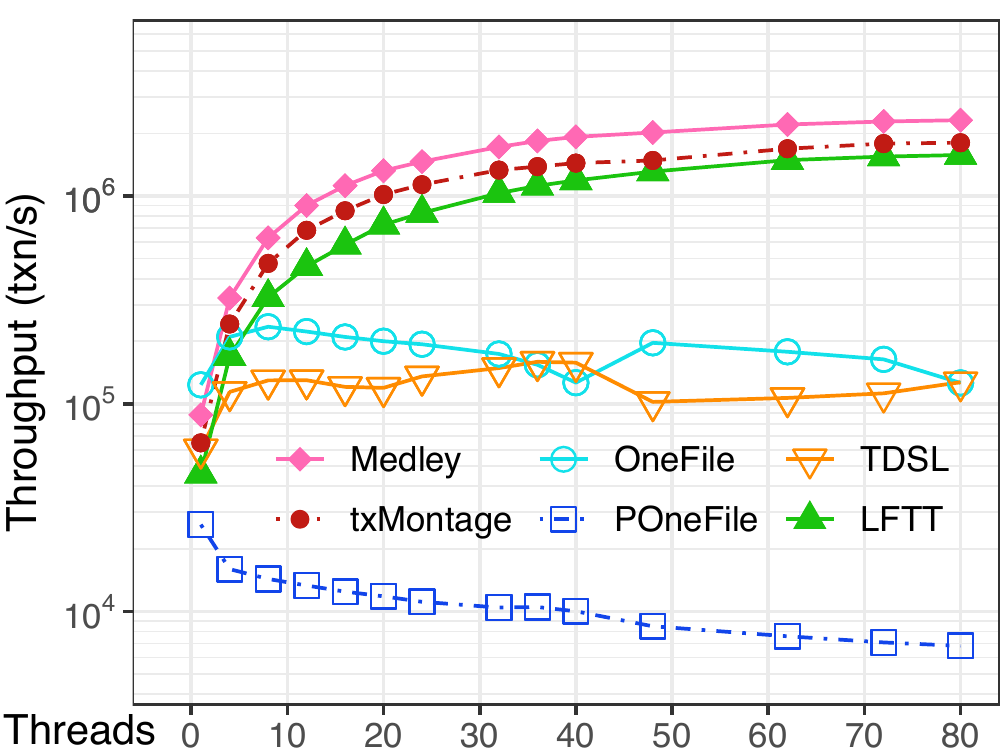}%
        \vspace{-1ex}%
        }
    \hfill
    \subfloat[get:insert:remove
    2:1:1]%
        {\includegraphics[width=.3\textwidth]%
            {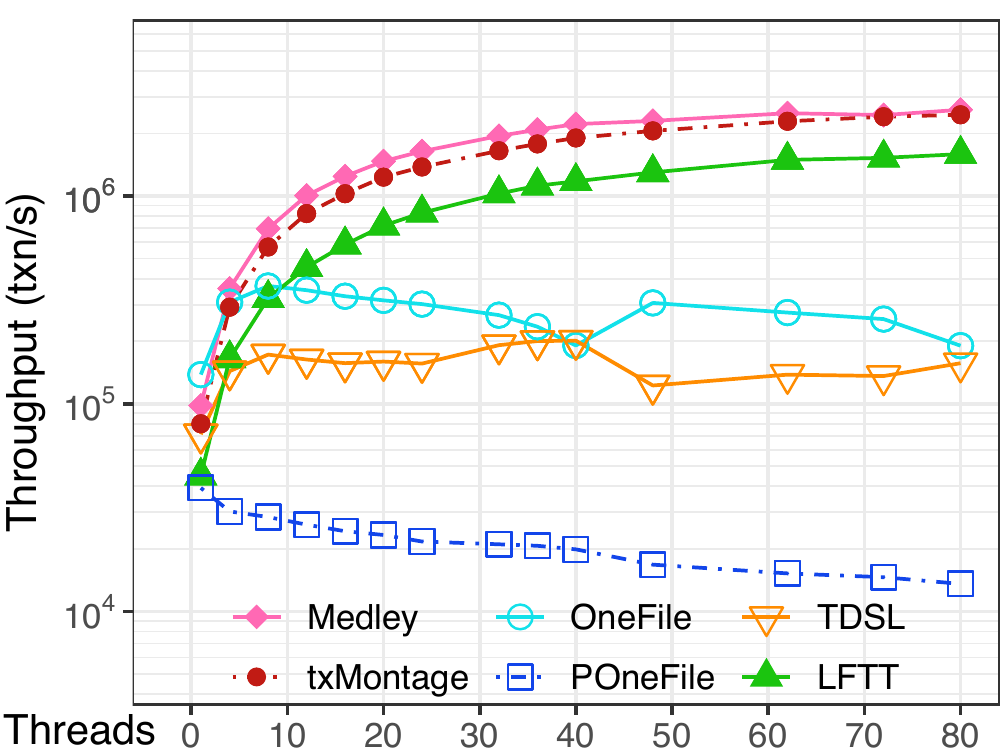}%
        \vspace{-1ex}%
        }
    \hfill
    \subfloat[get:insert:remove
    18:1:1]%
        {\includegraphics[width=.3\textwidth]%
            {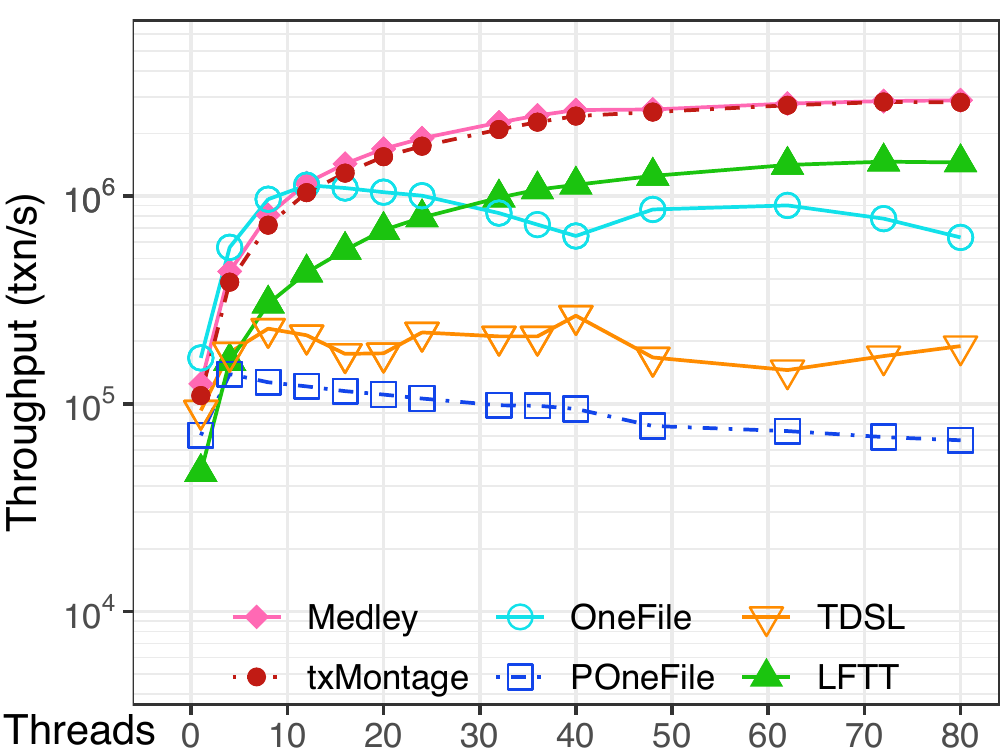}
        \vspace{-1ex}
        }
    \hfill\strut
    \vspace{-2ex}%
    \caption{Throughput of transactional skiplists (log $Y$
      axis).}
    \vspace{-1ex}%
    \label{fig:throughput-sls}
\end{figure*}
\fi

Throughput results for the hash table and skiplist microbenchmarks
appear in Figures~\ref{fig:throughput-hts} and~\ref{fig:throughput-sls},
respectively.
Solid lines represent transactions on transient data structures;
dotted lines represent persistent transactions.
Considering only the transient case for now, Medley consistently
outperforms the transient version of OneFile by more than an order of
magnitude, on both hash tables and skiplists, for anything more than a
trivial number of threads.
The gap becomes larger when the workload has a higher percentage of
writes.  Despite its lack of scalability, OneFile performs well at small
thread counts,
especially with a read-mostly workload.
We attribute this fact to its serialized transaction design, which
eliminates the need for read sets.

As described in Section~\ref{sec:intro}, TDSL provides (blocking)
transactions over various specially constructed data structures.  While
conflicts still occur on writes, read sets are limited to only
semantically critical nodes, and the authors report significant
improvements in throughput relative to general-purpose
STM~\cite{spiegelman-2016-pldi}.  As shown in
Figure~\ref{fig:throughput-sls}, however, TDSL, like OneFile, has
limited scalability, and is dramatically outperformed by Medley.
Somewhat to our surprise, TDSL also fails to outperform OneFile on this
microbenchmark, presumably because of the latter's elimination
of read sets.

Among the various skiplist competitors, LFTT comes closest to rivaling
Medley, but still trails by a factor of 1.4--2$\times$ in the
\iftighten\else best
(write-only)\fi\ case.  Re-executing entire transactions in LFTT introduces
\iftighten\else considerable\fi\ redundant work---planning in particular.
On read-mostly workloads, where Medley benefits from invisible readers,
LFTT trails by a factor of 2--2.7$\times$.

As a somewhat more realistic benchmark, we repeated our comparison of
Medley, OneFile, and TDSL on the \code{newOrder} and \code{payment}
transactions of TPC-C\@.
We were unable to include LFTT in these tests because it supports only
static transactions, in which the set of data structure operations is
known in advance---nor could we integrate its dynamic variant
(DTT~\cite{laborde-pmam-2019}), as the available version of the code
does not allow arbitrary key and value types.
\citet{laborde-pmam-2019} report, however, that DTT's performance
is similar to that of LFTT on simple transactions.  Given that DTT has
to publish the entire transaction as a lambda expression on all its
critical nodes, we would expect DTT's performance to be, if anything,
somewhat worse on the large transactions of TPC-C, and LFTT was already
about 2$\times$ slower than Medley on the micro\-benchmark.


TPC-C throughput for Medley, (transient) OneFile, and TDSL appears in
Figure~\ref{fig:tpcc}.  Because transactions on TPC-C
are large, OneFile is impacted severely.  By ensuring the atomicity of
only critical accesses, Medley still scales for large numbers of threads
and outperforms the competition by as much as 45$\times$.

\subsection{Latency (Transient)}
\label{sec:latency}

\ifarxiv
\begin{figure}
    \centering
    \includegraphics[width=2in,height=1.6in]%
        {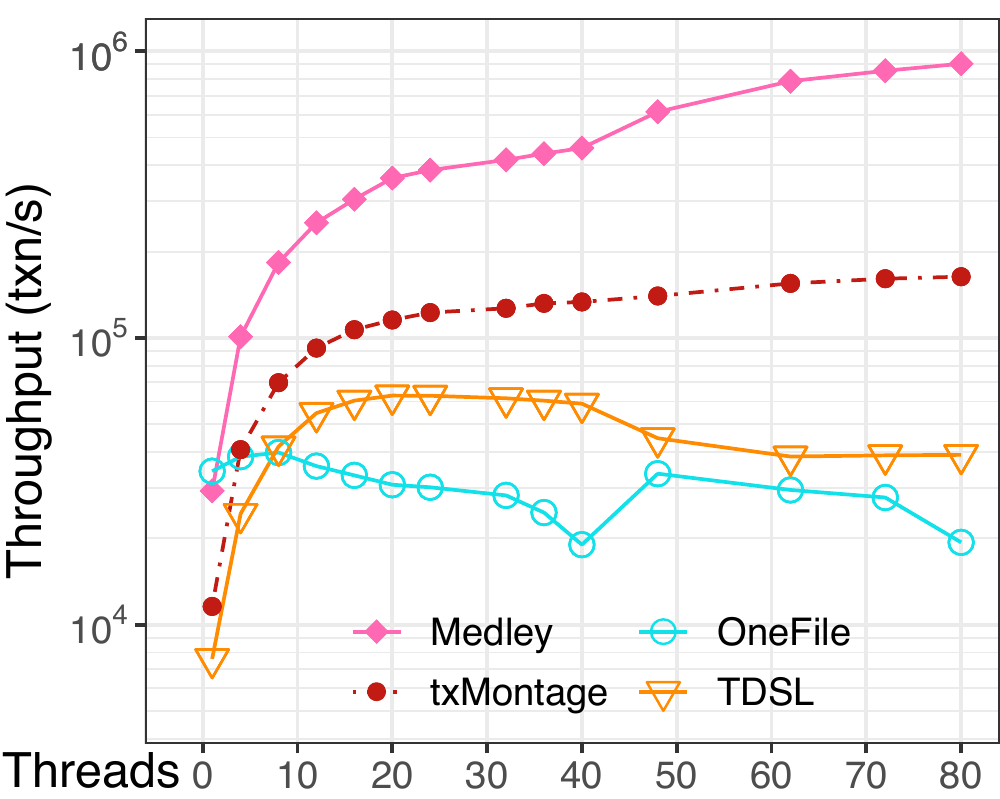}
    \caption{TPC-C skiplist throughput (log $Y$ axis).}
    \vspace{-1ex}
    \label{fig:tpcc}
\end{figure}

\begin{figure}
    \subfloat[on DRAM]%
        {\includegraphics[height=1.45in]%
            {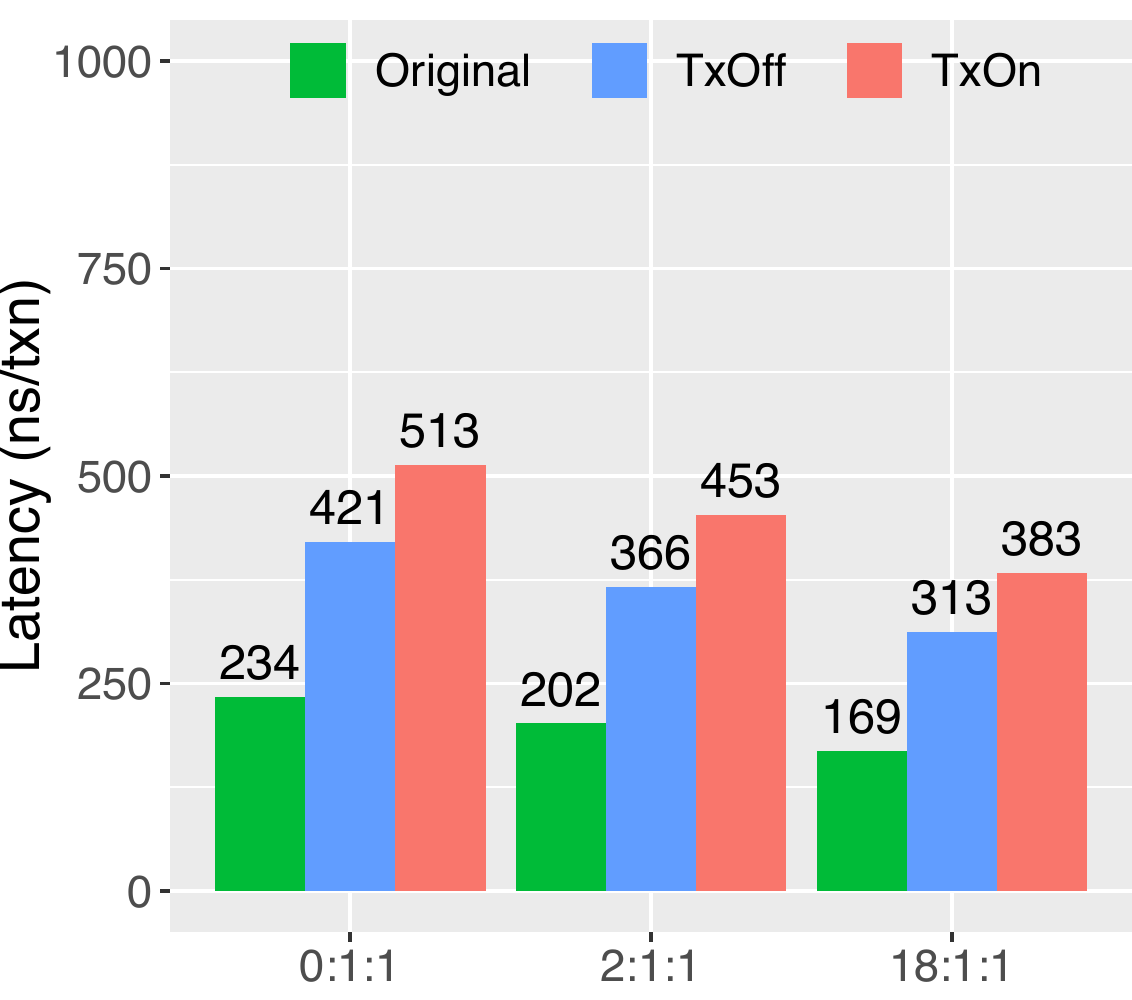} \label{fig:latency_dram}
        \vspace{-1ex}%
        }
    \subfloat[transient on NVM]%
        {\includegraphics[height=1.45in]%
            {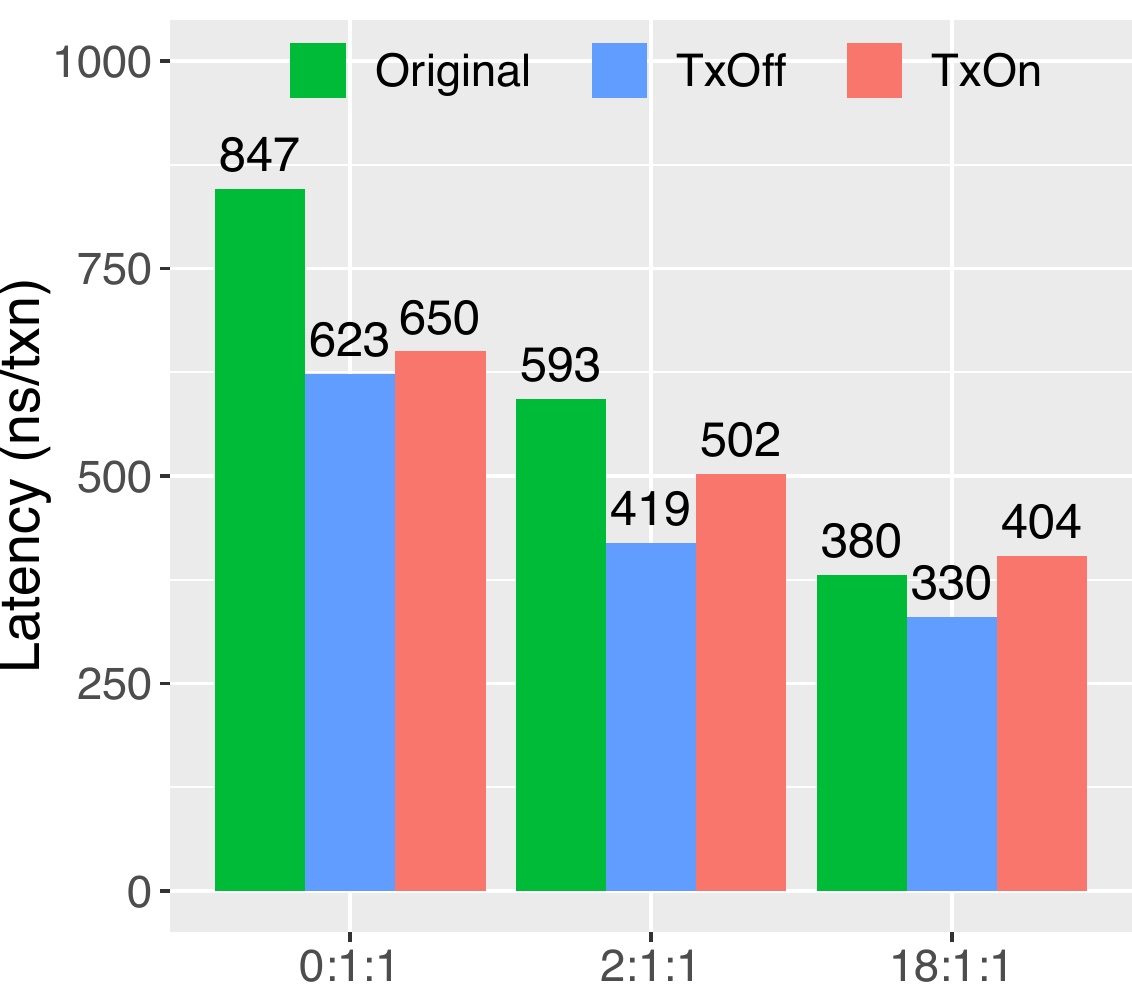}\label{fig:latency_nvm}
        \vspace{-1ex}%
        }
    \hspace{2pt}%
    \subfloat[persistent on NVM]%
        {\includegraphics[height=1.45in]%
            {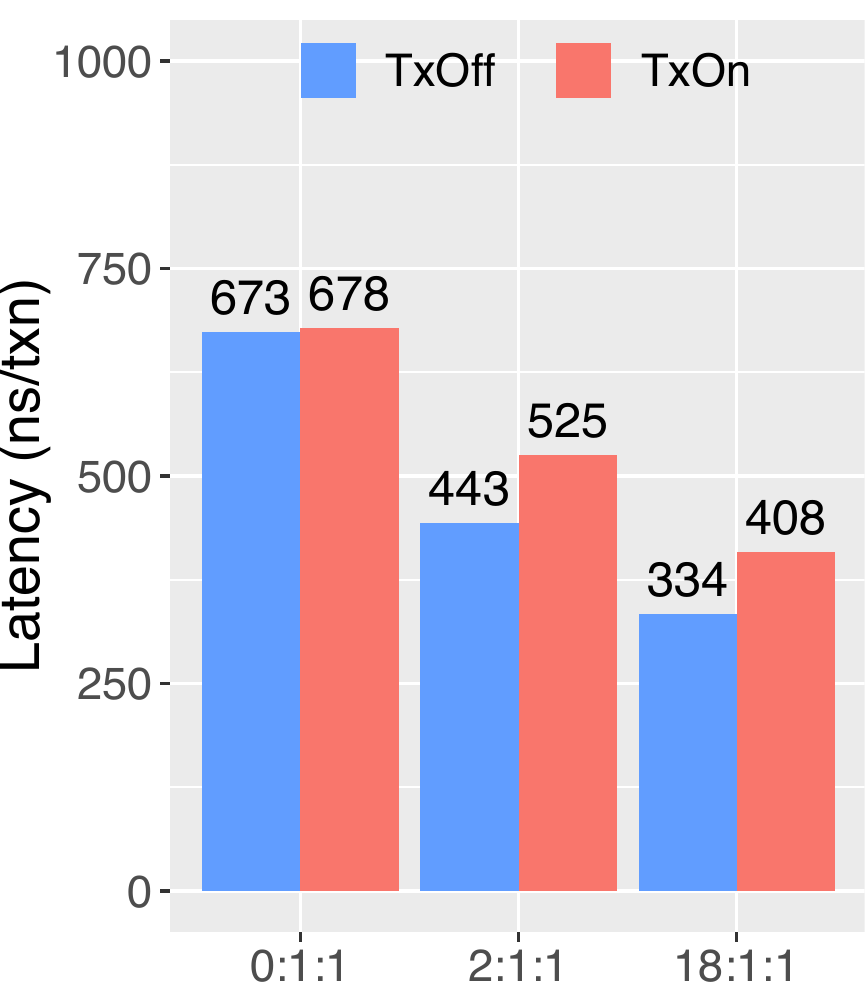}\label{fig:latency_p}
        \vspace{-1ex}
        }
    \hfill\strut
    \vspace{-2ex}%
    \captionsetup{justification=centering}
    \caption{Average latency on skiplists at 40 threads.\\
     $X$ labels are ratio of get:insert:remove.}
    \vspace{-1ex}%
    \label{fig:latency-sls}
\end{figure}
\else
\begin{figure*}
\begin{minipage}[b]{.28\textwidth}
    \includegraphics[width=2in,height=1.6in]%
        {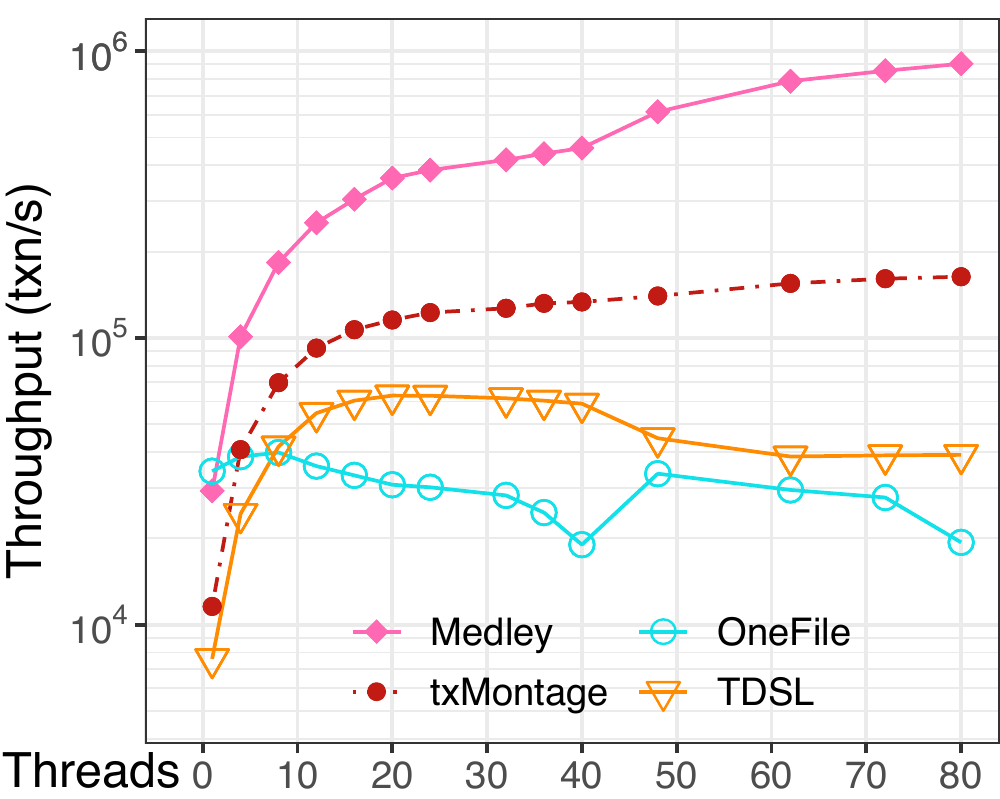}
    \vspace{-3ex}
    \captionsetup{justification=centering}
    \caption{TPC-C skiplist throughput (log $Y$ axis).}
    \vspace{-1ex}
    \label{fig:tpcc}
\end{minipage}
\hfill
\addtocounter{figure}{1}    
\begin{minipage}[b]{.7\textwidth}
    \strut\hfill
    \subfloat[on DRAM]%
        {\includegraphics[height=1.5in]%
            {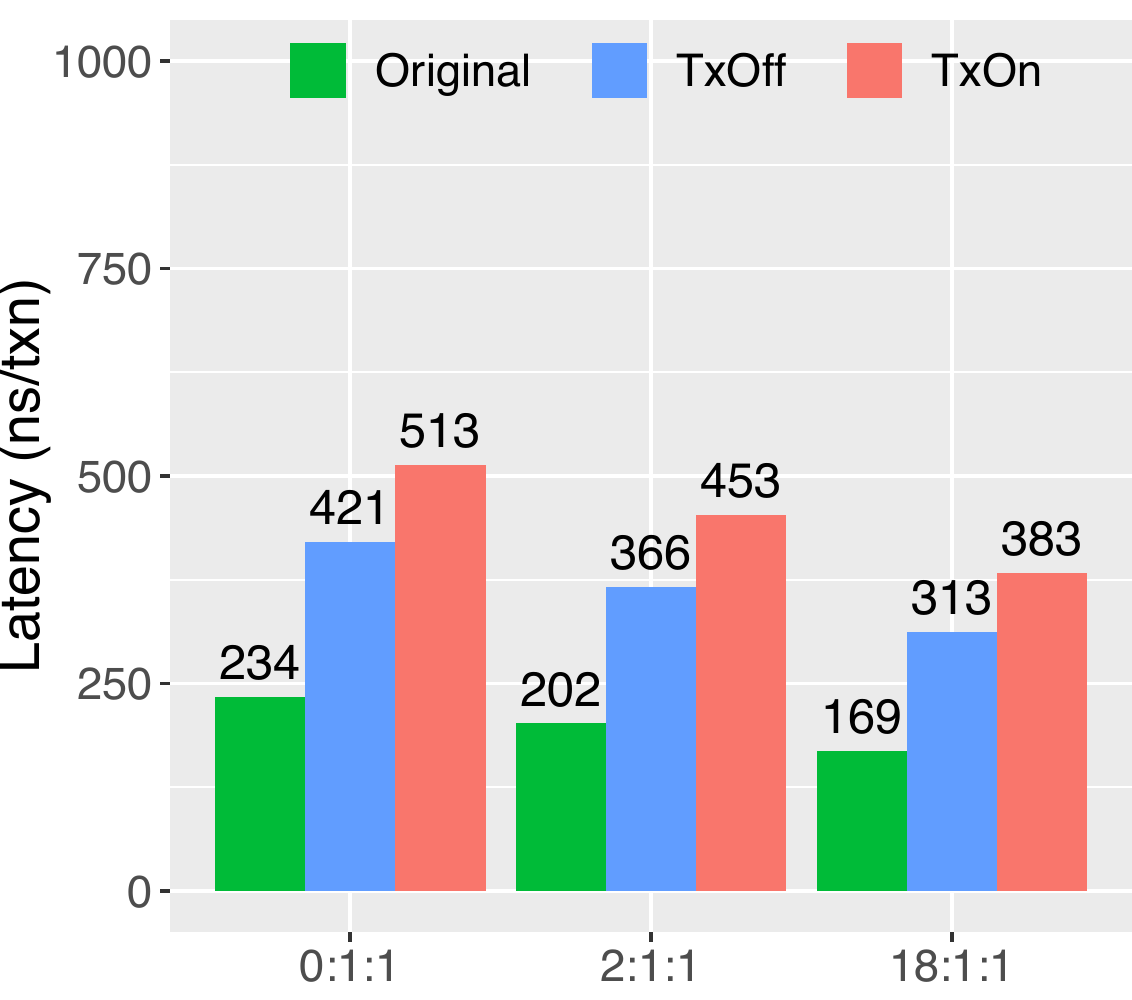} \label{fig:latency_dram}
        \vspace{-1ex}%
        }
    \subfloat[transient on NVM]%
        {\includegraphics[height=1.5in]%
            {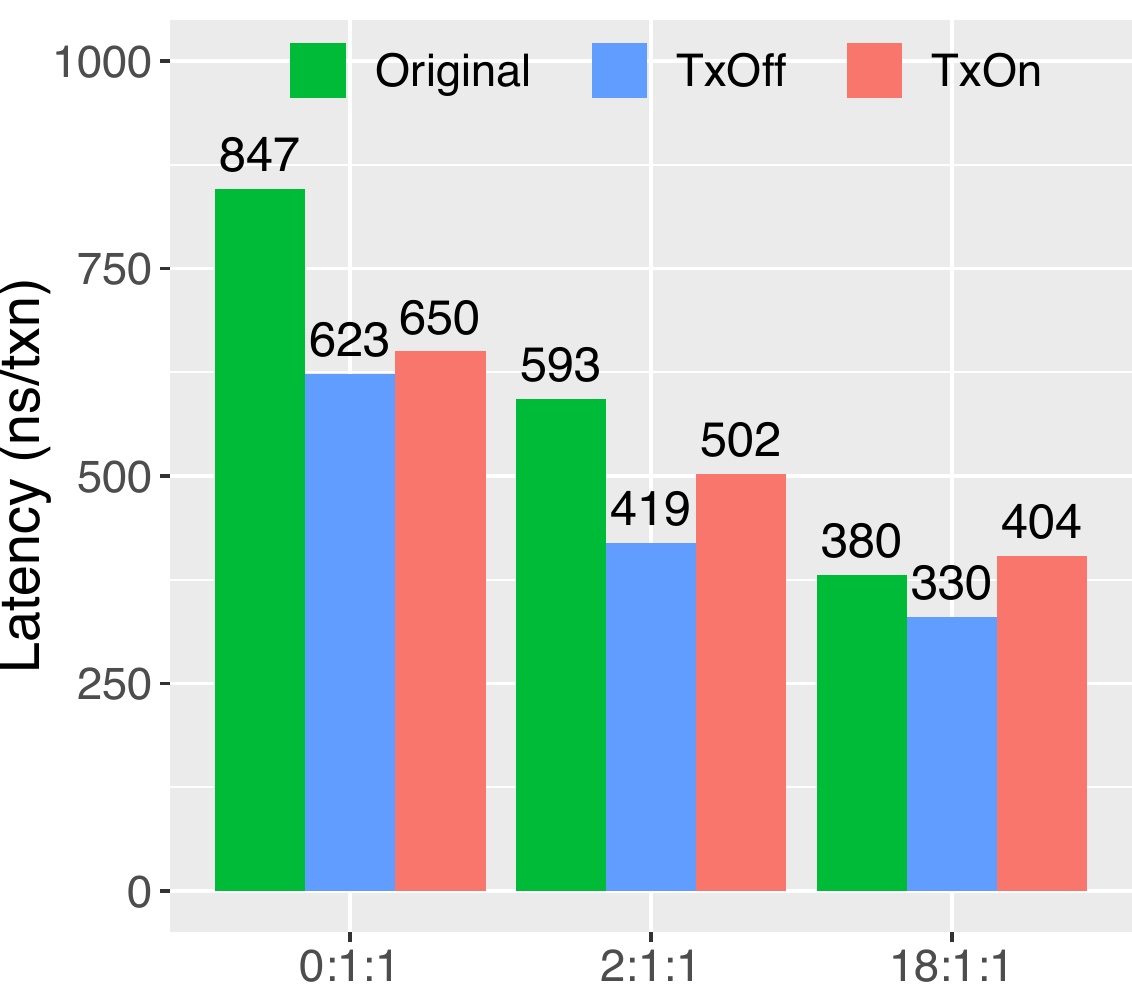}\label{fig:latency_nvm}
        \vspace{-1ex}%
        }
    \hspace{2pt}%
    \subfloat[persistent on NVM]%
        {\includegraphics[height=1.5in]%
            {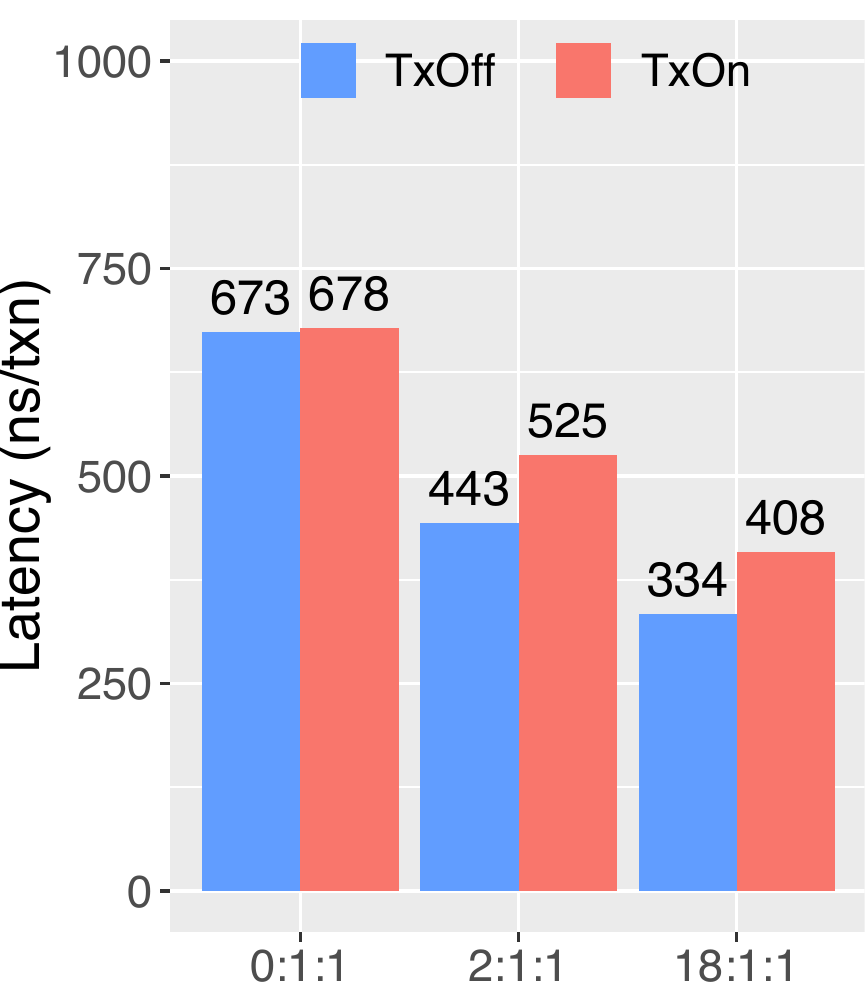}\label{fig:latency_p}
        \vspace{-1ex}
        }
    \hfill\strut
    \vspace{-2ex}%
\addtocounter{figure}{-1}   
    \captionsetup{justification=centering}
    \caption{Average latency on skiplists at 40 threads.\\
     $X$ labels are ratio of get:insert:remove.}
    \vspace{-1ex}%
    \label{fig:latency-sls}
\end{minipage}
\end{figure*}
\fi

In an attempt to assess the marginal cost of transaction composition,
we re-ran our microbenchmark on Fraser's original skiplist
(\textbf{Original}---no transactions), the NBTC-trans\-formed skip\-list
without transactions (\textbf{TxOff}---no calls to
\code{txBegin} or \code{txEnd}), and the NBTC-transformed skip\-list
with transactions (\textbf{TxOn}---as in
Figure~\ref{fig:throughput-sls}).

Figure~\ref{fig:latency_dram} reports latency for structures
placed in DRAM\@.  Without transactions, the transformed skip\-list
is 1.8$\times$ slower than the original.  With transactions turned on,
it's about 2.2$\times$ slower.
These results suggest that the more-than-doubled cost of CASes
(installing and uninstalling descriptors) accounts for about
$^2\kern-1.5pt/\kern-1pt_3$ of Medley's overhead.

\subsection{Persistence}
\label{sec:pm_exp}

To evaluate the impact of failure atomicity and durability on the
throughput of txMontage, we can return to the dotted lines of
Figures~\ref{fig:throughput-hts}, \ref{fig:throughput-sls},
and~\ref{fig:tpcc}.
%

\subsubsection*{Throughput}

In the microbenchmark tests,
with strict persistence and eager cache-line write-back, persistent
OneFile is an order of magnitude slower than its transient version.
With periodic persistence, however, the txMontage hash table achieves
half the throughput of Medley at 40 threads on the write-only
workload---almost two orders of magnitude faster than POneFile.
With a read-mostly workload on the hash table, or with any of the
workloads on the skiplist (with its lower overall concurrency), txMontage
is almost as fast as Medley.
In the extreme write-heavy case (80 threads on the 0:1:1 hash table
workload), we attribute the roughly 4$\times$ slowdown of txMontage
to NVM's write bottleneck~\cite{izraelevitz-optane-2019}---in
particular, to the phenomenon of
\emph{write amplification}~\cite{wang-micro-2020,hu-sigmod-2022}.

Results are similar in TPC-C (Fig.~\ref{fig:tpcc}).
Transactions here are both large and heavy on writes;
allocating payloads on NVM limits txMontage's throughput to roughly a
fifth of Medley's, but that is still about 4$\times$ faster than
transient OneFile.
POneFile, for its part, spent so long on the warm-up phase of TPC-C that
we lost patience and killed the test.

\subsubsection*{Latency}

Figure~\ref{fig:latency_nvm} shows the latency of skiplist transactions
when txMontage payloads are allocated on NVM (and indices on DRAM) but
persistence is turned off (no epochs or explicit cache line write-back).
For comparison, we have also shown the latency of the original,
non-transactional skiplist with \emph{all} data placed in NVM\@.
Figure~\ref{fig:latency_p} shows the corresponding latencies for
fully operational txMontage.

Comparing Figures~\ref{fig:latency_dram} and~\ref{fig:latency_nvm}, we
see lower marginal overhead for transactions when running on NVM.
This may suggest that the hardware write bottleneck is reducing overall
throughput and thus contention.

On the write-only workload (leftmost groups of bars), moving payloads to NVM
introduces an overhead of almost 50\% (Fig.~\ref{fig:latency_dram}
versus Fig.~\ref{fig:latency_nvm}).
On the read-mostly workload (rightmost bars), this penalty drops
to 5\%. Again, we attribute the effect to NVM's write bottleneck.  The high
latency of the original skiplist entirely allocated on NVM (green bars
in Figure~\ref{fig:latency_nvm}) appears to confirm this hypothesis.

Comparing Figures~\ref{fig:latency_nvm} and~\ref{fig:latency_p},
txMontage pays less than 5\%, relative to Medley on NVM, for failure atomicity
and durability.

\section{Conclusion}
\label{sec:conclusion}

We have presented nonblocking transaction composition (NBTC), a new
methodology that leverages the linearizability of existing nonblocking
data structures when building dynamic transactions.  As concrete
realizations, we introduced the \emph{Medley} system for transient
structures and the \emph{txMontage} system for (buffered) persistent
structures.
Medley transactions are isolated and consistent; txMontage transactions
are also failure atomic and durable.  Both systems are quite fast: where
even the best STM has traditionally suffered slowdowns of
3--10$\times$, Medley incurs more like 2.2$\times$;
txMontage, for its part, adds only 5--20\% to the overhead of nbMontage,
allowing it to outperform existing nonblocking persistent STM systems
by nearly two orders of magnitude.


Given their eager contention management, Medley and txMontage
maintain obstruction freedom for transactions on nonblocking
structures.  In future work, we plan to explore lazy contention
management, postponing installment of descriptors until transactions
are ready to commit.
\iftighten
\else By sorting and installing descriptors in canonical
order, the resulting systems would
preserve lock freedom.  Lazy
contention management would also facilitate helping, as any installed
descriptor would have \code{status == InProg}, and any other thread
could push it to completion.

As currently defined in NBTC, \emph{speculation intervals} are easy to
identify, but may unnecessarily instrument certain harmless helping
instructions between publication and linearization.
We are currently working to develop a more precise but still tractable
definition of helping in order to
\fi
reduce the number of ``critical''
memory accesses that must be performed atomically in each transaction.



\bibliography{arxiv}


\end{document}